\RequirePackage{amsmath}
\documentclass[runningheads]{llncs}
\usepackage{multicol}
\usepackage{amsfonts}
\usepackage{amssymb}
\usepackage{graphicx}
\usepackage{epic}
\usepackage{eepic}
\usepackage{epsfig,float}
\usepackage{verbatim}
\usepackage{pdfsync}
\usepackage{gastex}
\usepackage[lowtilde]{url}
\usepackage{geometry}
\usepackage{enumitem}
\usepackage{url} 
\usepackage{mathtools}
\usepackage{xcolor}
\usepackage{bbm}

\usepackage{ifpdf}
\ifpdf
\usepackage{pgf}
\usepackage{tikz}
\usepackage{todonotes}
\fi

\newcommand{\tge}{\trianglerighteq}
\newcommand{\tle}{\trianglelefteq}
\newcommand{\tg}{\triangleright}
\newcommand{\tl}{\triangleleft}

\usepackage{xcolor}

\usepackage[utf8]{inputenc}\usepackage[T1]{fontenc}

\renewcommand{\le}{\leqslant}
\renewcommand{\ge}{\geqslant}

\newcommand{\eps}{\varepsilon}

\newcommand{\Sig}{\Sigma}

\newcommand{\noin}{\noindent}

\newcommand{\bi}{\begin{itemize}}
\newcommand{\ei}{\end{itemize}}
\newcommand{\be}{\begin{enumerate}}
\newcommand{\ee}{\end{enumerate}}
\newcommand{\bd}{\begin{description}}
\newcommand{\ed}{\end{description}}
\newcommand{\bq}{\begin{quote}}
\newcommand{\eq}{\end{quote}}

\newcommand{\cD}{{\mathcal D}}

\newcommand{\cM}{{\mathcal M}}
\newcommand{\cN}{{\mathcal N}}

\newcommand{\cR}{{\mathcal R}}
\newcommand{\cS}{{\mathcal S}}

\newcommand{\one}{\mathbbm{1}}

\newcommand{\lraL}{{\mathbin{\approx_L}}}

\newenvironment{rethm}[2]{\noindent\textbf{#1 (Restated).}~\emph{#2}}

\DeclareMathOperator{\im}{im}

\title{Complexity of Proper Suffix-Convex Regular Languages
\thanks{This work was supported by the Natural Sciences and Engineering Research Council of Canada grant No.~OGP0000871, NSERC Discovery grant No.~8237-2012, and the Canada Research Chairs Program.}
}

\author{Corwin Sinnamon}
\authorrunning{C. Sinnamon}  
\titlerunning{Complexity of Proper Suffix-Convex Languages}

\institute{David R. Cheriton School of Computer Science, University of Waterloo, \\
Waterloo, ON, Canada N2L 3G1\\
{\tt sinncore@gmail.com}
}

\begin{document}
\maketitle
\begin{abstract}
A language $L$ is suffix-convex if for any words $u, v, w$, whenever $w$ and $uvw$ are in $L$, $vw$ is in $L$ as well.
Suffix-convex languages include left ideals, suffix-closed languages, and suffix-free languages, which were studied previously. 
In this paper, we concentrate on suffix-convex languages that do not belong to any one of these classes; we call such languages \emph{proper}.
In order to study this language class, we define a structure called a \emph{suffix-convex triple system} that characterizes the automata recognizing suffix-convex languages.
We find tight upper bounds for reversal, star, product, and boolean operations of proper suffix-convex languages, and we conjecture on the size of the largest syntactic semigroup.
We also prove that three witness streams are required to meet all these bounds.
\medskip

\noin
{\bf Keywords:}
atom, most complex, suffix-convex, proper, quotient complexity, regular language, state complexity, syntactic semigroup
\end{abstract}

\section{Introduction} 
\noin{\bf Suffix-Convex Languages:}
Convex languages were introduced in 1973 by Thierrin~\cite{Thi73}, and revisited in 2009 by Ang and Brzozowski~\cite{AnBr09}.
Convexity can be defined with respect to several binary relations on words, but in this paper we concentrate only on suffix-convex regular languages.
If a word $w \in \Sigma^*$ can be written as $w=xy$ for $x, y \in \Sigma^*$, then $y$ is a \emph{suffix} of $w$.
A language $L$ is \emph{suffix-convex} if whenever $w$ and $uvw$ are in $L$,  then $vw$ is also in $L$, for all $u, v, w \in \Sig^*$.
The class of suffix-convex languages includes three well-known subclasses: left ideals, suffix-closed languages, and suffix-free languages.

A language $L$ over an alphabet $\Sigma$ is a \emph{left ideal} if it is non-empty and $L=\Sigma^*L$.
In other words, if $L$ contains a word $w \in \Sigma^*$, then it also contains every word in $\Sigma^*$ that has $w$ as a suffix.
Left ideals play a role in pattern matching: If one is searching for all words ending with words in some language $L$ in a given text, then one is looking for words in $\Sig^*L$.
Left ideals also constitute a basic concept in semigroup theory.

A language $L$ is \emph{suffix-closed} if every suffix of every word in $L$ is also in $L$.
The complement of every suffix-closed language (other than $\Sigma^*$) is a left ideal.

A language is \emph{suffix-free} if no word in the language is a suffix of another word in the language.
Suffix-free languages are suffix codes (with the exception of $\{\eps\}$, where $\eps$ is the empty word).
They play an important role in coding theory and have been studied extensively; see~\cite{BPR10} for example. 

\noin{\bf Contributions:} In this paper, we focus on the remaining suffix-languages that do not fall into any of these subclasses; we call these languages \emph{proper}.
These languages are wide-ranging in structure and appearance, and difficult to reason about using conventional methods.
In order to approach the complexity properties of proper languages, we develop a theory of suffix-convex regular languages
based on a new object we call a suffix-convex triple system.
We use this theory to discover and prove tight upper bounds for reversal, star, product, and boolean operations of proper languages.
We describe a proper language that we conjecture to have the largest possible syntactic semigroup.
Finally, we prove that three different language streams are required to meet all of these bounds.

\section{Background}

\noin{\bf Quotient/State Complexity:} If $L$ is a language over an alphabet $\Sig^*$, such that every letter of $\Sig$ appears in a word of $L$, then the \emph{(left) quotient} of $L$ by a word $w\in\Sig^*$ is $w^{-1}L=\{x\mid wx\in L\}$. A language is regular if and only if the set of distinct quotients is finite. For this reason the number of quotients of $L$ is a natural measure of complexity for $L$; this number is called the 
\emph{quotient complexity}~\cite{Brz10} of $L$.
A equivalent concept is the \emph{state complexity}~\cite{Yu01} of $L$, which is the number of states in a complete minimal deterministic finite automaton (DFA) over alphabet $\Sig$ recognizing $L$.
We refer to quotient/state complexity simply as \emph{complexity} and we denote it by $\kappa(L)$.

If $\circ$ is a unary operation on languages, then the \emph{quotient/state complexity of $\circ$} is the maximal value of $\kappa(L_n^\circ)$, expressed as a function of $n$, as $L_n$ ranges over all regular languages of complexity $n$ or less.
Similarly, if $\circ$ is a binary operation on languages, then the \emph{quotient/state complexity of $\circ$} is  the maximal value of $\kappa(L'_m \circ L_n)$, expressed as a function of $m$ and $n$, as $L'_m$ and $L_n$ range over all regular languages of complexity $m$ and $n$, respectively.
We assume in this paper that $L'_m$ and $L_n$ are over a common alphabet $\Sigma$, however the \emph{unrestricted} complexity of binary operations, where the two languages may use different alphabets, has recently been studied as well~\cite{Brz16,BrSi17}.
The complexity of an operation gives a worst-case bound on the time and space complexity of the operation, and it has been studied extensively (see~\cite{Brz10,Brz13,GMRY16,HoKu11,Yu01}).
\smallskip

\noin
{\bf Witness Streams:}
To find the complexity of a unary operation one proves an upper bound on this complexity and then exhibits languages that meet this bound. Since a bound is given as a function of $n$, we require a sequence of languages $(L_k, L_{k+1}, \dots)$ called a language \emph{stream}; here $k$ is usually a small integer because the bound may not hold for a few small values of $n$.
Usually the languages in a stream have the same basic structure and differ only in the parameter $n$. For example, $((a^n)^* \mid n\ge 2)$ is a stream.
Two streams are required for a binary operation.
Sometimes the same stream can be used for both arguments, however this is not the case in general.
\smallskip

\noin{\bf Dialects:}
It has been shown in~\cite{Brz13} that for all common binary operations on regular languages the second stream can be a ``dialect'' of the first.
Let $\Sigma=\{a_1, \dots, a_k\}$ be an alphabet ordered as shown;
if $L\subseteq \Sigma^*$, we denote it by $L(a_1,\dots,a_k)$.
A \emph{dialect}  of $L$ is obtained by changing or deleting letters of $\Sigma$ in the words of $L$.
More precisely, if $\Sigma'$ is an alphabet, a dialect of $L(a_1, \dots, a_k)$ is obtained from an injective partial map $\pi \colon \Sigma \mapsto \Sigma'$ by replacing each letter $a \in \Sigma$ by $\pi(a)$ in every word of $L$,
or deleting the word entirely if $\pi(a)$ is undefined.
We write $L(\pi(a_1),\dots, \pi(a_k))$ to denote the dialect of $L(a_1,\dots,a_k)$ given by $\pi$,
and we denote undefined values of $\pi$ by  ``$-$''.
Undefined values for letters at the end of the alphabet are omitted; for example, $L(a,c,-,-)$ is written as $L(a,c)$.
\smallskip

\noin{\bf Automata:}
A \emph{deterministic finite automaton (DFA)} is a quintuple
$\cD=(Q, \Sigma, \delta, q_0,F)$, where
$Q$ is a finite non-empty set of \emph{states},
$\Sig$ is a finite non-empty \emph{alphabet},
$\delta\colon Q\times \Sig\to Q$ is the \emph{transition function},
$q_0\in Q$ is the \emph{initial} state, and
$F\subseteq Q$ is the set of \emph{final} states.
We extend $\delta$ to a function $\delta\colon Q\times \Sig^*\to Q$ as usual.
A DFA $\cD$ \emph{accepts} a word $w \in \Sigma^*$ if and only if $\delta(q_0,w) \in F$.
The language of all words accepted by $\cD$ is denoted $L(\cD)$.
If $q$ is a state of $\cD$, then the language of $q$ is the language accepted by the DFA $(Q,\Sigma,\delta,q,F)$.
The language of $q$ is a quotient of $L(\cD)$, and we often denote it $K_q$.
A state is \emph{empty} or a \emph{sink state} if its language is empty.
Two states $p$ and $q$ of $\cD$ are \emph{equivalent} if $K_p = K_q$; otherwise they are \emph{distinguishable}. 
A state $q$ is \emph{reachable} if there exists $w\in\Sig^*$ such that $\delta(q_0,w)=q$.
A DFA is \emph{minimal} if all of its states are reachable and no two states are equivalent.
Usually DFAs are used to establish upper bounds on the complexity of operations and also as witnesses that meet these bounds.
For convenience, say that a DFA is (proper) suffix-convex if the language it accepts is (proper) suffix-convex.

A \emph{nondeterministic finite automaton (NFA)} is a quintuple
$\cD=(Q, \Sigma, \delta, I,F)$, where
$Q$,
$\Sig$ and $F$ are defined as in a DFA, 
$\delta\colon Q\times \Sig\to 2^Q$ is the \emph{transition function}, and
$I\subseteq Q$ is the \emph{set of initial states}. 
An \emph{$\eps$-NFA} is an NFA in which transitions under the empty word $\eps$ are also permitted.

\noin{\bf Transformations:}
Without loss of generality we take $Q_n=\{0,\dots,n-1\}$ to be the states set of every DFA with $n$ states.
A \emph{transformation} of $Q_n$ is a function $t \colon Q_n\to Q_n$.
We treat a transformation $t$ as an operator acting on $Q_n$ from the right, so that $qt$ denotes the \emph{image} of $q\in Q_n$ under $t$.
If $s, t$ are transformations of $Q_n$, their composition is denoted by $s \circ t$, or more commonly just $s t$, and defined by $q(st)=(qs)t$.
In any DFA, each letter $a\in \Sig$ induces a transformation $\delta_a$ of the set $Q_n$ defined by $q\delta_a=\delta(q,a)$. 
By a slight abuse of notation, we use the letter $a$ to denote the transformation it induces; thus we write $qa$ instead of $q\delta_a$.
We also extend the notation to sets of states: if $P\subseteq Q_n$, then $Pa=\{pa\mid p\in P\}$.
Alternatively, we write $P\xrightarrow{a} P'$ to indicate that the image of $P$ under $a$ is $P'$.

For $k\ge 2$, a transformation $t$ of a set $P=\{q_0,q_1,\dots,q_{k-1}\} \subseteq Q_n$ is called a \emph{$k$-cycle}
if $q_0t=q_1, q_1t=q_2,\dots,q_{k-2}t=q_{k-1},q_{k-1}t=q_0$, and we denote such a cycle by $(q_0,q_1,\dots,q_{k-1})$.
A 2-cycle $(q_0,q_1)$ is called a \emph{transposition}.
A transformation is a called a \emph{permutation} if it is bijective, or equivalently, if it can be written as a composition of cycles.
A transformation  that sends all the states of $P$ to $q$ and acts as the identity on the remaining states is denoted by $(P \to q)$.
 If $P=\{p\}$ we write  $p\to q$ for $(\{p\} \to q)$.
 The identity transformation is denoted by $\one$.
 The notation $(_i^j \; q\to q+1)$ denotes a transformation that sends $q$ to $q+1$ for $i\le q\le j$ and acts as the identity for the remaining states.
The notation $(_i^j \; q\to q-1)$ is defined similarly. Using composition, the notation introduced here lets us succinctly describe many different transformations.

\noin{\bf Semigroups:}
Let $\cD = (Q_n, \Sig, \delta, q_0, F)$ be a DFA.
For each word $w \in \Sig^*$, the transition function induces a transformation $\delta_w$ of $Q_n$ by  $w$: for all $q \in Q_n$, $q\delta_w = \delta(q, w)$.
The set $T_{\cD}$ of all such transformations by non-empty words forms a semigroup of transformations called the \emph{transition semigroup} of $\cD$~\cite{Pin97}. 
Conversely, we may define $\delta$ by describing $\delta_a$ for each $a \in \Sig$.
We write $a\colon t$, where $t$ is a transformation of $Q$, to mean that the transformation $\delta_a$ induced by $a$ is $t$. 

The \emph{Myhill congruence}~\cite{Myh57} $\lraL$ of a language $L\subseteq \Sig^*$ is defined on $\Sig^+$ as follows:
For $x, y \in \Sig^+$, $x \lraL y$ if and only if $wxz\in L \iff wyz\in L$ for all  $w,z \in\Sig^*$.
This relation is also known as the \emph{syntactic congruence} of $L$.
The quotient set $\Sig^+/ \lraL$ of equivalence classes of the relation $\lraL$ is a semigroup called the \emph{syntactic semigroup} of $L$.
If  $\cD$ is a minimal DFA for $L$, then $T_{\cD}$ is isomorphic to the syntactic semigroup $T_L$ of $L$~\cite{Pin97}, and we represent elements of $T_L$ by transformations in~$T_{\cD}$. 
The \emph{syntactic complexity} of a language is the size of its syntactic semigroup, and it has been used as a measure of complexity for regular languages~\cite{Brz13,BrYe11,HoKo04}. 

\noin{\bf Atoms:}
Atoms are defined by the congruence in which two words $x$ and $y$ are equivalent if 
 $ux\in L$ if and only if  $uy\in L$ for all $u\in \Sig^*$. 
In other words, $x$ and $y$ are equivalent if $x\in u^{-1}L$ if and only if $y\in u^{-1}L$.
 An equivalence class of this congruence is called an \emph{atom} of $L$~\cite{BrTa14}.
 Thus, an atom is a non-empty intersection of complemented and uncomplemented quotients of $L$,
written $A_S = \bigcap_{i \in S} K_i \cap \bigcap_{i \not\in S} \overline{K_i}$ for $S \subseteq Q_n$, where $K_0, K_1, \dots, K_{n-1}$ are the quotients of $L$.
The number of atoms and the complexities of the atoms were suggested as measures of complexity of regular languages~\cite{Brz13}.
For more information about atoms and their complexity, see~\cite{BrTa13,BrTa14,Iva16}.


\section{Suffix-Convex Triple Systems}

Suffix-convex languages are difficult to reason about through the common representations of regular languages.
To alleviate this, we introduce a structure called a \emph{suffix-convex triple system} (or just ``triple system'').
A triple system is a set of 3-tuples of states in $Q_n$ that satisfy some structural conditions.
For every triple system, there is a nonempty family of DFAs on the state set $Q_n$ that are said to \emph{respect} the system.
Triple systems have the following properties:
\begin{enumerate}[topsep=0pt]
\item Every DFA that respects any triple system is suffix-convex.
\item For every suffix-convex DFA, there is at least one triple system that it respects.
\item For any triple system, among the transition semigroups of DFAs that respect the system, there is a unique maximal semigroup that contains all others.
\end{enumerate}
Through properties 1 and 2, triple systems effectively characterize suffix-convex regular languages.
However, as suggested by 2, the correspondence between triple systems and suffix-convex DFAs is not a bijection; 
most suffix-convex DFAs respect a number of different triple systems, and most triple systems are respected by many different DFAs.
Property 3 helps to identify DFAs of suffix-convex languages whose transition semigroups are particularly complex, which is useful both for discovering and reasoning about complex suffix-convex languages.

The inspiration for the triple system framework lies in the following reformulation of the definition of suffix-convexity.
A regular language $L$ is suffix-convex if and only if, for all $u, v, w \in \Sigma^*$, \[w^{-1}L \cap (uvw)^{-1}L \subseteq (vw)^{-1}L.\]
This statement is more usefully expressed in terms of the states of a DFA.
Let $\cD = (Q_n, \Sigma, \delta, 0, F)$ be a DFA, and let $K_q$ denote the language accepted by $(Q_n, \Sigma, \delta, q, F)$.
Setting $p = 0w$, $q = 0uvw$, and $r = 0vw$, the statement above becomes $K_p \cap K_q \subseteq K_r$.
This relationship between quotients satisfies some nice properties: If $p, q, r, s \in Q_n$, then
\begin{itemize}[topsep=0pt]
\item $K_p \cap K_q \subseteq K_p$,
\item $K_p \cap K_q \subseteq K_r \iff K_q \cap K_p \subseteq K_r$,
\item $K_p \cap K_q \subseteq K_r$ and $K_q \cap K_r \subseteq K_s \implies K_p \cap K_q \subseteq K_s$, and
\item $K_p \cap K_q \subseteq K_r$ and $p, q \in F \implies r \in F$.
\end{itemize}

All four properties are trivial to prove, yet it turns out that they capture the essential character of suffix-convex DFAs.
We can now make a formal definition, in which these four properties appear in a more abstract way.

\begin{definition}
A \emph{suffix-convex triple system} is a tuple $\cS = (Q, q_0, F, \cR)$, where $q_0 \in Q$ is the \emph{initial} state, $F \subseteq Q$ is a set of \emph{final} states, and $\cR \subseteq Q \times Q \times Q$ is a \emph{relation} such that, for all $p, q, r, s \in Q$,
\begin{enumerate}[leftmargin=*, labelindent=\widthof{(A)}]
\item[(A)] $(p, q, p) \in \cR$,
\item[(B)] $(p, q, r) \in \cR \iff (q, p, r) \in \cR$,
\item[(C)] $(p, q, r) \in \cR$ and $(q, r, s) \in \cR \implies (p, q, s) \in \cR$, and
\item[(D)] $(p, q, r) \in \cR$ and $p, q \in F \implies r \in F$.
\end{enumerate}
\end{definition}

\begin{definition}
A DFA $\cD = (Q, \Sigma, \delta, q_0, F)$ is said to \emph{respect} a triple system $\cS = (Q, q_0, F, \cR)$ if, for all transformations $t \in T_{\cD}$ and states $p, q, r \in Q$, it satisfies both
\begin{enumerate}[leftmargin=*, labelindent=\widthof{Condition 1:}]
\item[Condition 1:] $(p, q, r) \in \cR \implies (pt, qt, rt) \in \cR$.
\item[Condition 2:] $(q_0, q, r) \in\cR \implies (q_0, qt,rt) \in \cR$.
\end{enumerate}
Also say a transformation $t \colon Q \to Q$ \emph{respects} $\cS$ if it satisfies Conditions 1 and 2 for $\cS$.
\end{definition}

We frequently refer back to these definitions. Henceforth, let (A), (B), (C), (D), Condition 1, and Condition 2 denote the properties in these two definitions.

Notice that if a DFA $\cD$ respects a triple system $\cS$, then they must have the same state set, initial state, and final states.
As shorthand, we sometimes refer to a triple system $\cS = (Q, q_0, F, \cR)$ only by $\cR$ when the other parameters are clear from context.
In particular, it suffices to say that a DFA $\cD = (Q, \Sigma, \delta, q_0, F)$ respects $\cR$, since the other pieces of the triple system must be $Q$, $q_0$, and $F$.
In all future DFAs and triple systems, we use $Q_n$ as the state set and $0$ as the initial state.

Although the motivation for the triples in $\cR$ are those satisfying $K_p \cap K_q \subseteq K_r$ in some DFA that respects the system, it is not generally the case that  $(p, q, r) \in \cR \iff K_p \cap K_q \subseteq K_r$.
We can only guarantee that $(p, q, r) \in \cR \implies K_p \cap K_q \subseteq K_r$; the proof of this is an easy exercise using (D) and Condition 1.
Let us now prove the essential properties of triple systems that we mentioned at the beginning of this section.

\begin{proposition}
\label{prop:respectimpsc}
If a DFA $\cD = (Q_n, \Sigma, \delta, 0, F)$ respects a triple system $\cR$ then $\cD$ is suffix-convex.
\end{proposition}
\begin{proof}
Let $u,v,w \in \Sigma^*$ such that $w, uvw \in L(\cD)$.
To prove suffix-convexity, we show that $vw \in L(\cD)$.
Observe the following:
\begin{itemize}
\item $(0,0u,0) \in \cR$ by (A),
\item $(0, 0uv, 0v) \in \cR$ by Condition 2,
\item $(0w, 0uvw, 0vw) \in \cR$ by Condition 1,
\item $0w, 0uvw \in F$ since $w, uvw \in L(\cD)$,
\item $0vw \in F$ by (D).
\end{itemize}
Hence $vw \in L(\cD)$. \qed
\end{proof}

\begin{proposition}
\label{prop:scimprespect}
If a minimal DFA $\cD = (Q_n, \Sigma, \delta, 0, F)$ is suffix-convex, then it respects the triple system $(Q_n, 0 , F, \cR)$ where
\[\cR = \{(p, q, r) \mid K_p \cap K_q \subseteq K_r\}.\]
\end{proposition}
\begin{proof}
It is easy to verify that $(Q_n, 0, F, \cR)$ is a triple system when $\cD$ is minimal.
We must check that every transformation in $T_{\cD}$ satisfies Condition 1 and Condition 2.

\noindent \textbf{Condition 1:}
Let $t \in T_{\cD}$ and suppose $(p,q,r) \in \cR$. We wish to show that $(pt,qt,rt) \in \cR$, or equivalently, $K_{pt} \cap K_{qt} \subseteq K_{rt}$.
Choose a word $w \in \Sigma^*$ that induces $t$ in $\cD$.
Since $K_p \cap K_q \subseteq K_r$, we have $w^{-1}(K_p \cap K_q) \subseteq w^{-1}K_r$.
Notice $w^{-1}(K_p \cap K_q) = w^{-1}K_p \cap w^{-1}K_q = K_{pt} \cap K_{qt}$ and $w^{-1}K_r = K_{rt}$.
Therefore $K_{pt} \cap K_{qt} \subseteq K_{rt}$.

\noindent \textbf{Condition 2:}
Let $t \in T_{\cD}$ and suppose $(0, q, r) \in \cR$.
Then $K_0 \cap K_q \subseteq K_r$, and we wish to show $K_0 \cap K_{qt} \subseteq K_{rt}$.
To a contradiction, suppose there exists a word $w \in (K_0 \cap K_{qt}) \setminus K_{rt}$.
Choose words $u, v \in\Sigma^*$ such that $0u = q$ and $v$ induces $t$ in $\cD$.
Since $w \in K_0 \cap K_{pt}$, both $w$ and $uvw$ must be in $L$.
But by Condition 1, $K_{0t} \cap K_{qt} \subseteq K_{rt}$, and  since $w \in K_{qt}$ and $w \not\in K_{rt}$, it follows that $w \not\in K_{0t}$.
As $K_{0t} = K_{0v} = v^{-1}L$, we have $vw \not\in L$,  contradicting the suffix-convexity of $L(\cD)$.
\qed
\end{proof}

\begin{proposition}
\label{prop:realizability}
Let $\cS = (Q_n, 0, F, \cR)$ be a triple system and define
\[T^* \coloneqq \{t \colon Q_n \to Q_n \mid \text{$t$ respects $\cS$}\}.\]
If a DFA $\cD$ respects $\cS$ then $T_{\cD} \subseteq T^*$. Moreover, there is a DFA $\cD'$ respecting $\cS$ with $T_{\cD'} = T^*$.
\end{proposition}
\begin{proof}
The first claim is obvious, since every transformation in $T_{\cD}$ must respect $\cS$.
For the second claim, we may simply choose $\cD' = (Q_n, \Sigma, \delta, 0 F)$ where $\Sigma$ and $\delta$ are defined by
$\Sigma = \{a_t \mid t \in T^*\}$ and $\delta(p, a_t) = pt$ for all $t \in T^*$.
Since $T^*$ is a semigroup under composition, $T_{\cD'} = T^*$.
\qed
\end{proof} 


While a triple system gives a ternary relation between states, it also yields an interesting binary relation that is very useful in describing triple systems and reasoning about them.
As suggested by the asymmetry in Condition 2, the initial state $q_0$ plays a special role in a triple system. 

\begin{definition}
Given a triple system $\cS = (Q, q_0, F, \cR)$, define $\tle_R$, a binary relation on $Q$, by 
\[p \tle_\cR q \iff (q_0, p, q) \in \cR.\]
\end{definition}
When $\cR$ is clear from context, we will simply write $\tle$ instead of $\tle_\cR$.
It turns out that $\tle$ is a kind of order relation called a \emph{preorder} (also called a \emph{quasiorder}).

\begin{proposition}
For any triple system $\cS = (Q_n, 0, F, \cR)$, $\tle_\cR$ is a preorder on $Q_n$; that is, it satisfies
\begin{enumerate}
\item $p \tle_\cR p$, \hfill(Reflexivity)
\item $p \tle_\cR q$ and $q \tle_\cR r \implies p \tle_\cR r$.\hfill(Transitivity)
\end{enumerate}
\end{proposition}
\begin{proof}
Reflexivity follows by (A) and (B). To prove transitivity, suppose $p \tle_\cR q$ and $q \tle_\cR r$. Then
\begin{itemize}[topsep=0pt]
\item $(0, p, q) \in \cR$ by assumption,
\item $(p, 0, q) \in \cR$ by (B),
\item $(0, q, r) \in \cR$ by assumption,
\item $(p, 0, r) \in \cR$ by (C),
\item $(0, p, r) \in \cR$ by (B).
\end{itemize}
Hence $p \tle_\cR r$.\qed
\end{proof}
A preorder is similar to a partial order, except that it does not require the antisymmetry property ($p \tle q$ and $q \tle p \implies p = q$).
It is not true that $\tle$ is always a partial order, since there may be states $p$ and $q$ where $p \tle q$ and $q \tle p$, but $p \not= q$;
such elements are called \emph{symmetric} and we write $p \sim q$. We also write $p \tl q$ to indicate $p \tle q$ but $q \not\tle p$.

We will find $\tle$ useful because triple systems can be complicated and varied, whereas $\tle$ has a more restricted structure.
Besides being a preorder, $\tle$ has the interesting property that $p \tle 0$ for all $p \in Q_n$ (since $(0, p, 0) \in \cR$ for all $p \in Q_n$ by (A)).
Thus, $0$ is always a maximum element of $\tle$. Note that there could be other elements, symmetric with 0, which are also maximum elements with respect to $\tle$.

The most pleasing feature of $\tle$ is that it gives us an intuitive way of restating Condition 2.
\begin{enumerate}[leftmargin=*, labelindent=\widthof{Condition 2:}]
\item[\emph{Condition 2:}] \emph{$t$ is monotone with respect to $\tle$.}
\end{enumerate}
In this context, $t$ being \emph{monotone} means that $p \tle q \implies pt \tle qt$.
We frequently use this property and the structure of $\tle$ as an entry point to reasoning about triple systems.
It is sometimes sufficient to consider only $\tle$ in proofs, ignoring the finer details of the triple system entirely.
As demonstrated by the next theorem, when $\tle$ is a partial order we can effectively ignore the rest of the triple system because every monotone transformation can be included in the transition semigroup without breaking suffix-convexity.
Since it is never harmful to have a larger semigroup for proving complexity properties, the cases where $\tle$ is a partial order are the simplest and most natural.

\begin{theorem}
\label{thm:mono}
Fix any partial order $\preceq$ on $Q_n$ in which $0$ is the maximum element.
Let $f \in Q_n$ and consider the triple system $\cS = (Q_n, 0, \{f\}, \cR)$, where
\[\cR = \{(p, q, r) \mid p \preceq r \preceq q \text{ or } q \preceq r \preceq p\}.\]
There exists a minimal suffix-convex DFA $\cD = (Q_n, \Sigma, \delta, 0, F)$ respecting $\cS$ such that
\[T_{\cD} = \{t \colon Q_n \to Q_n \mid \text{ $t$ is monotone with respect to $\preceq$}\}.\]
Furthermore, $\tle_\cR = \preceq$, i.e. $p \tle_\cR q$ if and only if $p \preceq q$.
\end{theorem}
\begin{proof}
It is easy to check that $\cS$ satisfies (A), (B), (C), and (D).
By construction, $(0, p, q) \in \cR$ if and only if $p \preceq q \preceq 0$. S
ince $0$ is the maximum element in $\preceq$, this implies $\tle_\cR = \preceq$.
We construct a minimal DFA respecting $\cS$ with every monotone function in its transition semigroup.

Let $\cM$ denote the set of monotone transformations on $Q_n$ with respect to $\preceq$.
Since monotonicity is preserved under composition, $\cM$ is a semigroup under composition.
Let $\cD = (Q_n, \Sigma, \delta, 0, \{f\})$ where $\Sigma = \{a_t \mid t \in \cM\}$ and $\delta(p, a_t) = pt$ for all $t \in \cM$ (in other words, include a dedicated letter in $\Sigma$ for each monotone transformation). Clearly $T_{\cD} = \cM$.

It is easy to show that $\cD$ is minimal: State $p$ is reached by the transformation $(Q_n \to p)$, and two states $p$ and $q$, $q \not\preceq p$, are distinguished by the monotone transformation $t$ defined by \[rt = \begin{cases} f \quad&\text{if } r \preceq p, \text{ and}\\ 0 &\text{otherwise.}\end{cases}\]

To prove suffix-convexity, we show that every transformation in $\cM$ respects $\cS$. Condition 2 is trivial, since $\preceq = \tle_\cR$.
For Condition 1, observe that if $p \preceq r \preceq q$ then $pt \preceq rt \preceq qt$ for all $t \in \cM$.
Hence $(p, q, r) \in \cR$ implies $(pt, qt, rt) \in \cR$ for all $t \in \cM$.\qed
\end{proof}
\begin{remark}
The set of final states in Theorem 1 need not be a singleton. We only require that $\emptyset \subsetneq F \subsetneq Q_n$ and that $F$ is \emph{convex} with respect to $\preceq$; that is, there cannot be states $f \preceq g \preceq h$ where $f,h \in F$ and $g \not\in F$.
\end{remark}


\section{Star, Product, and Boolean Operations}
\label{sec:star}
This section has our first application of the triple system framework.
We present a proper suffix-convex witness stream $(L_n \mid n \ge 3)$ that meets the regular language upper bound for (Kleene) star.
With a dialect stream, it also meets the regular language upper bound for product and boolean operations.
Upper bounds for all of these operations on regular languages are well known (e.g. \cite{Brz13,Yu01}): If $L'$ and $L$ are regular languages of complexity $m$ and $n$, respectively, then $\kappa(L^*) \le 2^{n-1} + 2^{n-2}$, $\kappa(L'L) \le (m-1)2^n + 2^{n-1}$, and $\kappa(L' \circ L) \le mn$ for $\circ \in \{\cup, \oplus, \setminus, \cap\}$.

The witness DFA we introduce respects a triple system such that $\tle_\cR$ is a total order on $Q_n$.
We define the triple system such that $0 \tg 1 \tg \cdots \tg n-2 \tg n-1$.

\begin{definition}
\label{def:starsystem}
For $n \ge  3$, define $\cS_n = (Q_n, 0, \{n-2\}, \cR_n)$ where
\[\cR_n = \{(p, q, r) \mid  p \ge r \ge q \text{ or }  q \ge r \ge p\}.\]
\end{definition}
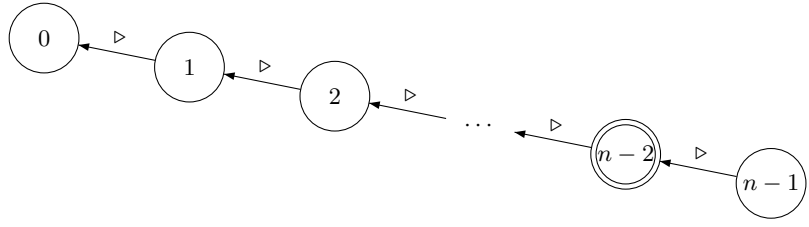
\begin{figure}
\unitlength 11pt
\begin{center}\begin{picture}(25,7)(0,-1)
\gasset{Nh=2.4,Nw=2.4,Nmr=1.2,ELdist=0.3,loopdiam=1.2}
\node(0)(0,5){$0$}
\node(1)(5,4){$1$}
\node(2)(10,3){$2$}
\node[Nframe=n](3)(15,2){$\cdots$}
\node(4)(20,1){$n-2$}\rmark(4)
\node(5)(25,0){$n-1$}

\drawedge[ELside=r](1,0){$\tg$}
\drawedge[ELside=r](2,1){$\tg$}
\drawedge[ELside=r](3,2){$\tg$}
\drawedge[ELside=r](4,3){$\tg$}
\drawedge[ELside=r](5,4){$\tg$}

\end{picture}\end{center}
\caption{The order relation $\tle_{\cR_n}$ of Definition~\ref{def:starsystem} used in the complex witness stream for star and product.}
\label{fig:starsystem}
\end{figure}

Note that $\cR_n$ is exactly the triple system from Theorem~\ref{thm:mono} if ``$\preceq$'' is replaced with ``$\ge$''.
Therefore, by Theorem 1, any monotone transformation can be included in the transition semigroup of the witness DFA without violating suffix-convexity.
For simplicity, we use a small alphabet that generates a non-maximal semigroup since it is sufficient for our purposes.

\begin{definition}
\label{def:starwitness}
For $n \ge 3$, let $L_n(\Sigma)$ be the language recognized by the DFA $\cD_n = (Q_n, \Sigma, \delta_n, 0, \{n-2\})$, where $\Sigma = \{a, b, c, d, e, f\}$ and $\delta_n$ is given by the transformations $a \colon (_0^{n-2} i \to i+1)$, $b \colon (_1^{n-1} i \to i-1)$, $c \colon (\{n-3, n-2\} \to n-1)$, $d \colon (n-2 \to n-1)$, and $e = f = \one$.
\end{definition}

\begin{figure}[h]
\unitlength 11pt
\begin{center}\begin{picture}(32,7)(-2,-1)
\gasset{Nh=2.4,Nw=2.4,Nmr=1.2,ELdist=0.3,loopdiam=1.2}
\node(0)(0,2){$0$}\imark(0)
\node(1)(4,2){$1$}
\node(2)(9,2){$2$}
\node[Nframe=n](3dots)(14,2){$\cdots$}
\node(n-3)(19,2){$n-3$}
\node(n-2)(24,2){$n-2$}\rmark(n-2)
\node(n-1)(29,2){$n-1$}

\drawedge[curvedepth=1](0,1){$a$}
\drawedge[curvedepth=1](1,2){$a$}
\drawedge[curvedepth=1](2,3dots){$a$}
\drawedge[curvedepth=1](3dots,n-3){$a$}
\drawedge[curvedepth=1](n-3,n-2){$a$}
\drawedge[curvedepth=1](n-2,n-1){$a, c, d$}

\drawedge[ELside=r,curvedepth=1](1,0){$b$}
\drawedge[ELside=r,curvedepth=1](2,1){$b$}
\drawedge[ELside=r,curvedepth=1](3dots,2){$b$}
\drawedge[ELside=r,curvedepth=1](n-3,3dots){$b$}
\drawedge[ELside=r,curvedepth=1](n-2,n-3){$b$}
\drawedge[ELside=r,curvedepth=1](n-1,n-2){$b$}

\drawedge[curvedepth=-3](n-3,n-1){$c$}
\drawloop(0){$b,c,d,e,f$}
\drawloop(1){$c,d,e,f$}
\drawloop(2){$c,d,e,f$}
\drawloop(n-3){$d,e,f$}
\drawloop(n-2){$e,f$}
\drawloop(n-1){$a,c,d,e,f$}
\end{picture}\end{center}
\caption{DFA $\cD_n$ of Definition~\ref{def:starwitness}.}
\label{fig:starwitness}
\end{figure}
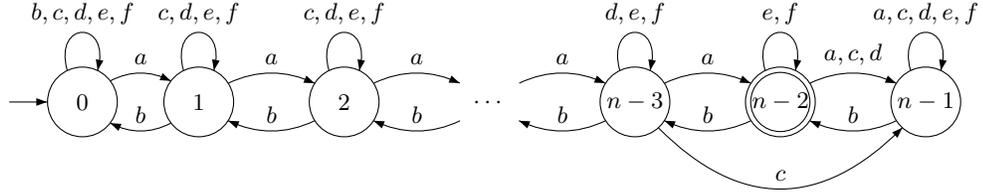

\begin{proposition}
\label{prop:starproper}
For $n \ge 3$, $L_n(\Sigma)$ of Definition~\ref{def:starwitness} is proper and $\kappa(L_n) = n$.
\end{proposition}
\begin{proof}
DFA $\cD_n(\Sigma)$ is minimal since $\cD_n(a)$ is minimal, and hence $\kappa(L_n) = n$.
By Theorem~\ref{thm:mono}, $L_n$ must be suffix-convex because the transformations induced by letters in $\Sigma$ are monotone with respect to the partial order $\tle_{\cS_n}$ of Definition~\ref{def:starsystem}.
It cannot be a left ideal  because $a^{n-2} \in L_n$ but $a^{n-1} \not\in L$, and hence $L_n \not= \Sigma^*L_n$.
It is not suffix-closed because $\eps \not\in L_n$.
Finally, it is not suffix-free because $ba^{n-2} \in L_n$ and $a^{n-2} \in L_n$. Thus $L_n$ is a proper language.\qed
\end{proof}

\begin{theorem}
\label{thm:star}
The language stream $(L_n(a, b, c , d) \mid n \ge 3)$ of Definition~\ref{def:starwitness} meets the upper bound for star.
That is, for $n \ge 3$, $\kappa(L_n^*) = 2^{n-1} + 2^{n-2}$.
\end{theorem}
\begin{proof}
We use the usual $\eps$-NFA construction for star.
Beginning with $\cD_n(a, b, c, d)$, create a new state called $0'$ which is final and has the same transitions as $0$. Set $0'$ to be the only initial state of the automaton, and add an $\eps$-transition from $n-2$ to $0$. The resulting automaton recognizes $L_n(a, b, c, d)^*$.
Applying the subset construction to this $\eps$-NFA, we show that there are $2^{n-1} + 2^{n-2}$ reachable and distinguishable subsets of $Q_n \cup \{0'\}$.
These subsets are $\{\{0'\}\} \cup \{P \subseteq Q_n \mid P \not= \emptyset \text{ and } n-2 \in P \implies 0 \in P\}$.

\begin{figure}[h]
\unitlength 11pt
\begin{center}\begin{picture}(32,8)(-2,-1.5)
\gasset{Nh=2.4,Nw=2.4,Nmr=1.2,ELdist=0.3,loopdiam=1.2}
\node(0')(0,6){$0'$}\imark(0')\rmark(0')
\node(0)(0,2){$0$}
\node(1)(4,2){$1$}
\node(2)(9,2){$2$}
\node[Nframe=n](3dots)(14,2){$\cdots$}
\node(n-3)(19,2){$n-3$}
\node(n-2)(24,2){$n-2$}\rmark(n-2)
\node(n-1)(29,2){$n-1$}

\drawedge[curvedepth=1](0',1){$a$}
\drawedge[ELside=r](0',0){$b,c,d$}

\drawedge[curvedepth=1](0,1){$a$}
\drawedge[curvedepth=1](1,2){$a$}
\drawedge[curvedepth=1](2,3dots){$a$}
\drawedge[curvedepth=1](3dots,n-3){$a$}
\drawedge[curvedepth=1](n-3,n-2){$a$}
\drawedge[curvedepth=1](n-2,n-1){$a, c, d$}

\drawedge[ELside=r,curvedepth=1](1,0){$b$}
\drawedge[ELside=r,curvedepth=1](2,1){$b$}
\drawedge[ELside=r,curvedepth=1](3dots,2){$b$}
\drawedge[ELside=r,curvedepth=1](n-3,3dots){$b$}
\drawedge[ELside=r,curvedepth=1](n-2,n-3){$b$}
\drawedge[ELside=r,curvedepth=1](n-1,n-2){$b$}

\drawedge[ELside=r,curvedepth=-3](n-3,n-1){$c$}
\drawedge[eyo=-0.8,exo=-0.05,ELside=r,curvedepth=4](n-2,0){$\eps$}
\end{picture}\end{center}
\caption{The $\eps$-NFA recognizing $L_n(a, b, c, d)^*$; missing transitions are self-loops.}
\label{fig:starnfa}
\end{figure}
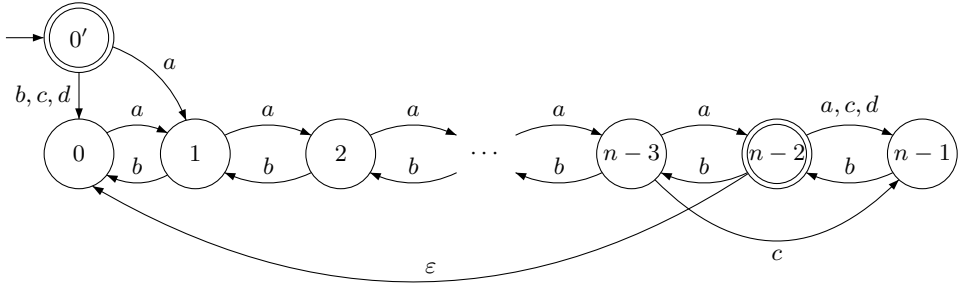

Obviously $\{0'\}$ is reachable because it is initial.
State $\{p\}$ for $0 \le p \le n-3$ is reachable from $\{0'\}$ by $ba^p$, and state $\{n-1\}$ is reached from $\{n-3\}$ by $c$.
State $\{0, n-2\}$ is reachable by $a^{n-2}$.
State $\{0, p_1, p_2, \dots, p_k, n-2\}$ with $0 < p_1 < p_2 < \cdots < p_k < n-2$ is reached from state $\{0, p_2-p_1, p_2, \dots, p_{k-1}-p_1, p_k - p_1, n-2\}$ by the word $a^{p_1+1}b$.
By induction, we can reach every state $P$ with $0, n-2 \in P$ and $n-1 \not\in P$.
We can then reach $P \cup \{n-1\}$: If $P = \{0, p_1, p_2, \dots, p_k, n-2\}$, then $P \cup \{n-1\}$ is reached from $\{0, p_1-1, p_2-1, \dots, p_k-1, n-3, n-2\}$ by $a$.

It remains to show reachability for states not containing $n-2$.
The state $\{p_1, p_2, \dots, p_k\}$ with $0 \le p_1 < p_2 < \cdots < p_k < n-2$ is reached from the state $\{0, p_1 + (n-2 - p_k), p_2  + (n-2 - p_k), \dots, p_{k-1} + (n-2 - p_k), n-2\}$ by $b^{p_1 + (n-2 - p_k)}a^{p_1}$.
Finally, state $\{p_1, p_2, \dots, p_k, n-1\}$ with $0 \le p_1 < p_2 < \cdots < p_k < n-2$ is reached from $\{0, p_2 - p_1, p_3 - p_1, \dots, p_k - p_1, n-2\}$ by $da^{p_1}$.
Thus, all $2^{n-1} + 2^{n-2}$ states are reachable.

State $\{0'\}$ can be distinguished from any state $P \subseteq Q_n$ by $\eps$ if $n-2 \not\in Q_n$, or by $ab$ if $n-2 \in Q_n$.
Any two states $P_1, P_2 \subseteq Q_n$ with $p \in P_1 \setminus P_2$ may be distinguished by $a^{n-2-p}$ if $p \not= n-1$, or by $b$ if $p = n-1$.
Hence, all states are pairwise distinguishable. Therefore $\kappa(L_n(a, b, c, d)^*) = 2^{n-1} + 2^{n-2}$; this is maximal for regular languages. \qed
\end{proof}

One may wonder what is required of a suffix-convex language to meet the bound for star.
It turns out that the triple system it respects must be somewhat similar to that of Definition~\ref{def:starsystem}.

\begin{lemma}
\label{lem:starreq}
Suppose $L$ is a suffix-convex language with $\kappa(L) = n \ge 3$ and $\kappa(L^*) = 2^{n-1} + 2^{n-2}$.
Let $\cD = (Q_n, \Sigma, \delta, 0, F)$ be a minimal DFA for $L$ and assume that $\cD$ respects a triple system $\cR$.
Then $\tle_\cR$ must admit a comparison between every pair of states in $Q_n$, i.e. for all  $p, q \in Q_n$, either $p \tle q$ or $q \tle p$.
\end{lemma}
This lemma does not imply that $\tle$ must be a total order on the states because it could have symmetric elements.
\begin{proof}
Perform the usual $\eps$-NFA construction for star on $\cD$, wherein a new initial and final state $0'$ is added to the DFA with the same transitions as $0$, and an $\eps$-transition is added from each $f \in F$ to $0$.
The only (potentially) reachable states in the subset construction of this $\eps$-NFA are
\[\text{$\{0'\}$, $\{S \subseteq Q_n \setminus F \mid S \not= \emptyset\}$, and $\{S \subseteq Q_n \mid 0 \in S, S \cap F \not= \emptyset\}$.}\]
In order for $L$ to meet the upper bound for star, $F$ must be a singleton $\{f\}$ and all of the above sets must be reachable.

To complete the proof, we show that for every reachable state $S \subseteq Q_n$, $\tle$ admits comparisons between every pair of states in $S$.
Since we know that every one of the above sets, including $Q_n$, is reachable, it will follow that all of the states are pairwise comparable with respect to $\tle$.

We proceed by induction on $|S|$. The statement is trivially true when $|S| = 1$. Assume now that $|S| \ge 2$.
Since the only way to reach larger sets in the subset construction is by using the $\eps$-transition from $f$ to $0$,
there must exist a set $S'$ of size $|S|-1$ and a transformations $t_1 \in T_{\cD}$ such that
$S'$ is reachable and $f \in (S't_1)$ so that $t_1$ maps $S'$ to $(S't_1) \cup \{0\}$ in the NFA.
Then, there must exist a transformation $t_2 \in T_{\cD}$ such that $(S't_1) \cup \{0\} \xrightarrow{t_2} S$.
As the elements of $S'$ are pairwise comparable (by inductive assumption) and $t_1$ is monotone, the elements of $(S't_1)$ must also be pairwise comparable.
The elements of $(S't_1) \cup \{0\}$ must also be pairwise comparable because $0$ is a maximum element in $\tle$.
Finally, since $t_2$ is monotone, we conclude that the elements of $S$ are pairwise comparable.
\qed
\end{proof}

For binary operations, two DFAs are considered at once. To avoid confusion between the two, we replace the state set of one of the DFAs with a ``primed'' version, in which each state $p$ becomes $p'$. The primed version of a DFA $\cD$ is denoted $\cD'$.

\begin{theorem}
\label{thm:product}
The dialect streams $(L_m(a, b, c , -, e, f) \mid m \ge 3)$ and $(L_n(e, f, -, -, a, b) \mid n \ge 3)$ of Definition~\ref{def:starwitness} meet the upper bound for product of proper suffix-convex languages.
Specifically, for $m, n \ge 3$, $\kappa(L_m(a, b, c , -, e, f)L_n(e, f, -, -, a, b)) = (m-1)2^n + 2^{n-1}$.
\end{theorem}
\begin{proof}
Let $\cD'_m = \cD'_m(a, b, c, -, e, f)$ and $\cD_n = \cD_n(e, f, -, -, a, b)$ be DFAs as defined in Definition~\ref{def:starwitness}.
To clearly distinguish between the two DFAs, we take the state set of $\cD'_m$ to be $Q'_m = \{0', 1', \dots, (m-1)'\}$.
To construct an $\eps$-NFA for the product $L_m(a, b, c, -, e, f)L_n(e, f, -, -, a, b)$, connect $\cD'_m$ and $\cD_n$ with an $\eps$-transition from $(m-2)'$ to $0$.
Let $0'$ be the initial state and let $n-2$ be the only final state.
We perform the subset construction on this NFA and show that the $(m-1)2^n + 2^{n-1}$ states $\{\{p'\} \cup S \mid p' \not= (m-2)', S \subseteq Q_n\} \cup \{\{(m-2)', 0\} \cup S \mid S \subseteq Q_n\}$ are reachable and pairwise distinguishable.

\begin{figure}[h]
\unitlength 11pt
\begin{center}\begin{picture}(34,11)(-1,-1)
\gasset{Nh=2.4,Nw=3.5,Nmr=1.2,ELdist=0.3,loopdiam=1.2}
{\small
\node(0')(0,6){$0'$}\imark(0')
\node(1')(5,6){$1'$}
\node[Nframe=n](3dots')(10,6){$\cdots$}
\node(m-3')(15,6){$(m-3)'$}
\node(m-2')(20,6){$(m-2)'$}
\node(m-1')(25,6){$(m-1)'$}

\drawedge[curvedepth=1](0',1'){$a$}
\drawedge[curvedepth=1](1',3dots'){$a$}
\drawedge[curvedepth=1](3dots',m-3'){$a$}
\drawedge[curvedepth=1](m-3',m-2'){$a$}
\drawedge[curvedepth=1](m-2',m-1'){$a,c$}

\drawedge[ELside=r,curvedepth=1](1',0'){$b$}
\drawedge[ELside=r,curvedepth=1](3dots',1'){$b$}
\drawedge[ELside=r,curvedepth=1](m-3',3dots'){$b$}
\drawedge[ELside=r,curvedepth=1](m-2',m-3'){$b$}
\drawedge[ELside=r,curvedepth=1](m-1',m-2'){$b$}
\drawedge[curvedepth=3.5](m-3',m-1'){$c$}

\node(0)(12,0){$0$}
\node(1)(17,0){$1$}
\node[Nframe=n](3dots)(22,0){$\cdots$}
\node(n-2)(27,0){$n-2$}\rmark(n-2)
\node(n-1)(32,0){$n-1$}

\drawedge[curvedepth=1](0,1){$e$}
\drawedge[curvedepth=1](1,3dots){$e$}
\drawedge[curvedepth=1](3dots,n-2){$e$}
\drawedge[curvedepth=1](n-2,n-1){$e$}

\drawedge[ELside=r,curvedepth=1](1,0){$f$}
\drawedge[ELside=r,curvedepth=1](3dots,1){$f$}
\drawedge[ELside=r,curvedepth=1](n-2,3dots){$f$}
\drawedge[ELside=r,curvedepth=1](n-1,n-2){$f$}

\drawbpedge[ELside=r](m-2',-90,7,0,90,7){$\eps$}
}
\end{picture}\end{center}
\caption{The $\eps$-NFA recognizing $L_m(a,b,c,-,e,f)L_n(e, f, -, -, a, b)$; missing transitions are self-loops.}
\label{fig:productNFA}
\end{figure}
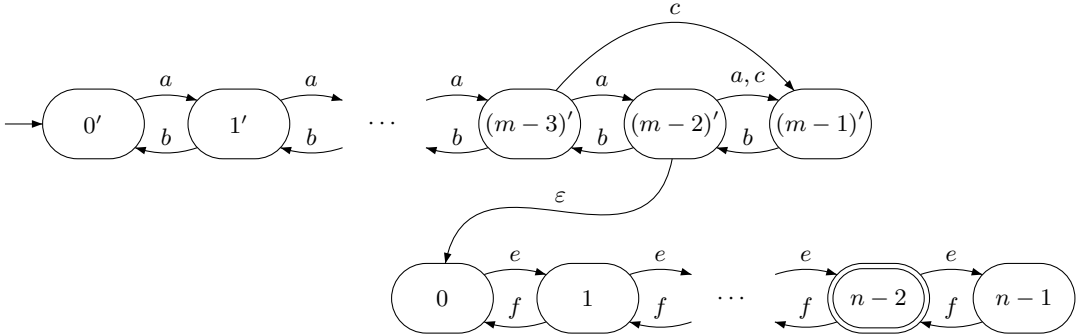

First, state $\{p'\}$ is reached by $a^p$ if $p < m-2$, and $\{(m-1)'\}$ is reached from $\{(m-3)'\}$ by $c$.
The state $\{(m-2)', 0\}$ is reached by $a^{m-2}$
For $k \ge 1$, state $\{(m-2)', 0, q_1, q_2, \dots, q_k\}$ with $0 < q_1 < q_2 < \cdots < q_k$ is reached from $\{(m-2)', 0, q_2 - q_1, q_3 - q_1, \dots, q_k - q_1\}$ by $ae^{q_1}b$; hence $\{(m-2)', 0\} \cup S$ is reachable for all $S \subseteq Q_n$.
If $p' \in \{0', \dots, (m-3)'\}$ and $S \subseteq Q_n$ then $\{p', 0\} \cup S$ is reached from $\{(m-2)', 0\} \cup S$ by $b^{m-2-p}$, and $\{(m-1)', 0\} \cup S$ is reached from $\{(m-2)', 0\} \cup S$ by $a$.
Finally, if $p' \not= (m-2)'$ and $0 < q_1 < q_2 < \cdots < q_k$ then $\{p', q_1, \dots, q_k\}$ is reached from $\{p', 0, q_2-q_1, \dots, q_k-q_1\}$ by $e^{q_1}$.
This proves that all the desired states are reachable.

Any two states $\{p'_1\} \cup S_1$ and $\{p'_2\} \cup S_2$ are distinguishable: If there is some $q \in S_1 \setminus S_2$ then they are distinguished by $e^{n-2-q}$ if $q \not= n-1$, or by $f$ if $q = n-1$.
Otherwise, assuming without loss of generality that $p'_1 < p'_2$, the two states are distinguished by $a^{m-2-p_1}e^n$. This proves that the product has complexity $(m-1)2^n + 2^{n-1}$.\qed
\end{proof}

\begin{theorem}
\label{thm:boolean}
The dialect streams $(L_m(a, b, -, -, e, f) \mid m \ge 3)$ and $(L_n(e, f, -, -, a, b) \mid n \ge 3)$ of Definition~\ref{def:starwitness} meet the upper bounds for boolean operations on proper suffix-convex languages.
Specifically, for $m, n \ge 3$ and $\circ \in \{\cup, \oplus, \setminus, \cap\}$, $\kappa(L_m(a, b, -, -, e, f) \circ L_n(e, f, -, -, a, b)) = mn$.
\end{theorem}
\begin{proof}
The common construction for boolean operations is the direct product of the DFAs.
For each operation $\circ \in \{\cup, \oplus, \setminus, \cap\}$, define a DFA $\cD^{\circ}$ with alphabet $\{a, b, e, f\}$, states $Q'_m \times Q_n$, initial state $(0', 0)$, and transitions given by $(p', q) \xrightarrow{w} (p'w, qw)$ for all $w \in \{a, b, e, f\}^*$.
The final states of $\cD^{\circ}$ are equal to $F^\circ = (\{(m-2)'\} \times Q_n) \circ (Q'_m \times \{n-2\})$.

\begin{figure}[h]
\unitlength 8pt
\begin{center}\begin{picture}(42,22)(-1,-1)
\gasset{Nh=2.4,Nw=7.5,Nmr=1.2,ELdist=0.3,loopdiam=1.2}
{\small
\node(0'0)(0,20){$0', 0$}\imark(0'0)
\node(1'0)(10,20){$1', 0$}
\node[Nframe=n](3'0)(20,20){$\cdots$}
\node(m-2'0)(30,20){$(m-2)', 0$}\rmark(m-2'0)
\node(m-1'0)(40,20){$(m-1)', 0$}

\node(0'1)(0,15){$0', 1$}
\node(1'1)(10,15){$1', 1$}
\node[Nframe=n](3'1)(20,15){$\cdots$}
\node(m-2'1)(30,15){$(m-2)', 1$}\rmark(m-2'1)
\node(m-1'1)(40,15){$(m-1)', 1$}

\node[Nframe=n](0'3)(0,10){$\vdots$}
\node[Nframe=n](1'3)(10,10){$\vdots$}
\node[Nframe=n](3'3)(20,10){$\ddots$}
\node[Nframe=n](m-2'3)(30,10){$\vdots$}
\node[Nframe=n](m-1'3)(40,10){$\vdots$}

\node(0'n-2)(0,5){$0', n-2$}\rmark(0'n-2)
\node(1'n-2)(10,5){$1', n-2$}\rmark(1'n-2)
\node[Nframe=n](3'n-2)(20,5){$\cdots$}
}
{\scriptsize
\node(m-2'n-2)(30,5){$(m-2)', n-2$}
\node(m-1'n-2)(40,5){$(m-1)', n-2$}\rmark(m-1'n-2)
}
{\small
\node(0'n-1)(0,0){$0', n-1$}
\node(1'n-1)(10,0){$1', n-1$}
\node[Nframe=n](3'n-1)(20,0){$\cdots$}
}
{\scriptsize
\node(m-2'n-1)(30,0){$(m-2)', n-1$}\rmark(m-2'n-1)
\node(m-1'n-1)(40,0){$(m-1)', n-1$}
}

\drawedge[curvedepth=1](0'0,1'0){$a$}
\drawedge[curvedepth=1](1'0,0'0){$b$}
\drawedge[curvedepth=1](1'0,3'0){$a$}
\drawedge[curvedepth=1](3'0,1'0){$b$}
\drawedge[curvedepth=1](3'0,m-2'0){$a$}
\drawedge[curvedepth=1](m-2'0,3'0){$b$}
\drawedge[curvedepth=1](m-2'0,m-1'0){$a$}
\drawedge[curvedepth=1](m-1'0,m-2'0){$b$}

\drawedge[curvedepth=1](0'1,1'1){$a$}
\drawedge[curvedepth=1](1'1,0'1){$b$}
\drawedge[curvedepth=1](1'1,3'1){$a$}
\drawedge[curvedepth=1](3'1,1'1){$b$}
\drawedge[curvedepth=1](3'1,m-2'1){$a$}
\drawedge[curvedepth=1](m-2'1,3'1){$b$}
\drawedge[curvedepth=1](m-2'1,m-1'1){$a$}
\drawedge[curvedepth=1](m-1'1,m-2'1){$b$}

\drawedge[curvedepth=1](0'n-2,1'n-2){$a$}
\drawedge[curvedepth=1](1'n-2,0'n-2){$b$}
\drawedge[curvedepth=1](1'n-2,3'n-2){$a$}
\drawedge[curvedepth=1](3'n-2,1'n-2){$b$}
\drawedge[curvedepth=1](3'n-2,m-2'n-2){$a$}
\drawedge[curvedepth=1](m-2'n-2,3'n-2){$b$}
\drawedge[curvedepth=1](m-2'n-2,m-1'n-2){$a$}
\drawedge[curvedepth=1](m-1'n-2,m-2'n-2){$b$}

\drawedge[curvedepth=1](0'n-1,1'n-1){$a$}
\drawedge[curvedepth=1](1'n-1,0'n-1){$b$}
\drawedge[curvedepth=1](1'n-1,3'n-1){$a$}
\drawedge[curvedepth=1](3'n-1,1'n-1){$b$}
\drawedge[curvedepth=1](3'n-1,m-2'n-1){$a$}
\drawedge[curvedepth=1](m-2'n-1,3'n-1){$b$}
\drawedge[curvedepth=1](m-2'n-1,m-1'n-1){$a$}
\drawedge[curvedepth=1](m-1'n-1,m-2'n-1){$b$}

\drawedge[curvedepth=1](0'0,0'1){$e$}
\drawedge[curvedepth=1](0'1,0'0){$f$}
\drawedge[curvedepth=1](0'1,0'3){$e$}
\drawedge[curvedepth=1](0'3,0'1){$f$}
\drawedge[curvedepth=1](0'3,0'n-2){$e$}
\drawedge[curvedepth=1](0'n-2,0'3){$f$}
\drawedge[curvedepth=1](0'n-2,0'n-1){$e$}
\drawedge[curvedepth=1](0'n-1,0'n-2){$f$}

\drawedge[curvedepth=1](1'0,1'1){$e$}
\drawedge[curvedepth=1](1'1,1'0){$f$}
\drawedge[curvedepth=1](1'1,1'3){$e$}
\drawedge[curvedepth=1](1'3,1'1){$f$}
\drawedge[curvedepth=1](1'3,1'n-2){$e$}
\drawedge[curvedepth=1](1'n-2,1'3){$f$}
\drawedge[curvedepth=1](1'n-2,1'n-1){$e$}
\drawedge[curvedepth=1](1'n-1,1'n-2){$f$}

\drawedge[curvedepth=1](m-2'0,m-2'1){$e$}
\drawedge[curvedepth=1](m-2'1,m-2'0){$f$}
\drawedge[curvedepth=1](m-2'1,m-2'3){$e$}
\drawedge[curvedepth=1](m-2'3,m-2'1){$f$}
\drawedge[curvedepth=1](m-2'3,m-2'n-2){$e$}
\drawedge[curvedepth=1](m-2'n-2,m-2'3){$f$}
\drawedge[curvedepth=1](m-2'n-2,m-2'n-1){$e$}
\drawedge[curvedepth=1](m-2'n-1,m-2'n-2){$f$}

\drawedge[curvedepth=1](m-1'0,m-1'1){$e$}
\drawedge[curvedepth=1](m-1'1,m-1'0){$f$}
\drawedge[curvedepth=1](m-1'1,m-1'3){$e$}
\drawedge[curvedepth=1](m-1'3,m-1'1){$f$}
\drawedge[curvedepth=1](m-1'3,m-1'n-2){$e$}
\drawedge[curvedepth=1](m-1'n-2,m-1'3){$f$}
\drawedge[curvedepth=1](m-1'n-2,m-1'n-1){$e$}
\drawedge[curvedepth=1](m-1'n-1,m-1'n-2){$f$}
\end{picture}\end{center}
\caption{The DFA $\cD^{\oplus}$ recognizing $L_m(a,b,-,-,e,f) \oplus L_n(e, f, -, -, a, b)$; missing transitions are self-loops. DFAs $\cD^\cup$, $\cD^\setminus$, and $\cD^\cap$ differ only in the final states.}
\label{fig:booleanNFA}
\end{figure}
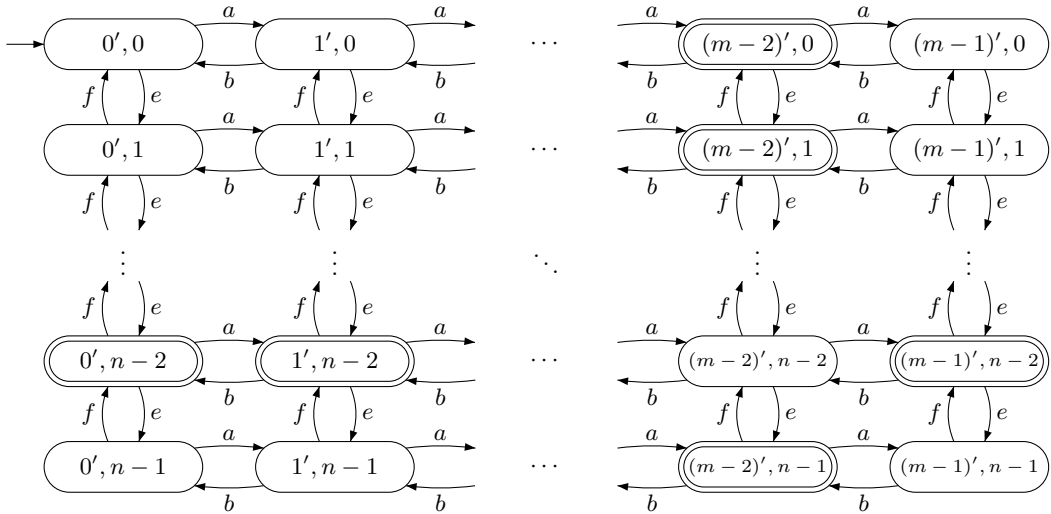

Reachability is the same for all operations: State $(p', q)$ is reached in $\cD^\circ$ by $a^pe^q$. We must show that states are pairwise distinguishable for each operation.

\noin\textbf{Union:} A state $(p', q)$ is final in $\cD^{\cup}$ if $p' = (m-2)'$ or $q = n-2$. State $(p'_1, q_1)$ is distinguishable from $(p'_2, q_2)$  in $\cD^\cup$ if $p_1 < p_2$ since $(p'_1, q_1) a^{m-2-p_1}e^{n-1} =  ((m-2)', n-1) \in F$ and $(p'_2, q_2) a^{m-2-p_1}e^{n-1} = ((m-1)', n-1) \not\in F$.
Similarly, if $q_1 < q_2$, then they are distinguished by $a^{m-1}e^{n-2-q_1}$ since this maps $(p'_1, q_1)$ to $((m-1)', n-2)$ and $(p'_2, q_2)$ to $((m-1)', n-1)$.

\noin\textbf{Symmetric Difference:} The final states in $\cD^{\oplus}$ are the same as in $\cD^{\cup}$ except for $((m-2)', n-2)$. The argument for union applies here as well.

\noin\textbf{Difference:} The final states in $\cD^{\setminus}$ are $((m-2)', q)$ for all $q \not= n-2$. If $p_1 < p_2$, $(p'_1, q_1)$ is distinguishable from $(p'_2, q_2)$  in $\cD^{\setminus}$ as in union.
If $q_1 < q_2$, then they are distinguished by $a^{m-1}be^{n-2-q_1}$ since this word maps $(p'_1, q_1)$ to $((m-2)', n-2) \not\in F$ and $(p'_2, q_2)$ to $((m-2)', n-1) \in F$.

\noin\textbf{Intersection:} The only final state in $\cD^{\cap}$ is $((m-2)', n-2)$. State $(p'_1, q_1)$ is distinguished from $(p'_2, q_2)$  in $\cD^\cap$ by $a^{m-2-p_1}e^{n-1}f$ if $p_1 < p_2$ or by $a^{m-1}be^{n-2-q_1}$ if $q_1 < q_2$.\qed
\end{proof}


\section{Reversal}
\label{sec:reversal}
We first prove an upper bound for the complexity of reversal in suffix-convex languages (not necessarily proper), and then give a proper suffix-convex witness stream that meets the bound for $n \ge 3$.

\begin{theorem}
\label{thm:reversalbd}
If $L$ is a suffix-convex language with $\kappa(L) = n$, then $\kappa(L^R) \le 2^n - 2^{n-3}$.
\end{theorem}
\begin{proof}
The case $n=1$ is trivial; assume $n \ge 2$.
We use the fact that the complexity of $L^R$ is equal to the number of atoms of $L$~\cite{BrTa14}.
We therefore wish to show that $L$ has at most $2^n - 2^{n-3} = \frac{7}{8}2^n$ atoms.

Let $\cD = (Q_n, \Sigma, \delta, 0, F)$ be a minimal DFA for $L$.
By Proposition~\ref{prop:scimprespect}, $\cD$ respects the triple system
\[\cR = \{(p, q, r) \mid K_p \cap K_q \subseteq K_r\}.\]
Recall that an atom is an intersection $A_S = \bigcap_{i \in S} K_i \cap \bigcap_{i \not\in S} \overline{K_i}$ for some $S \subseteq Q_n$, where $K_q$ denotes the language of state $q$ in $\cD$.
The number of atoms is the number of distinct, non-empty sets $A_S$ as $S$ ranges over the subsets of $Q_n$.

First, consider the case where there is some triple $(p, q, r) \in \cR$ with $r \not\in\{p, q\}$.
Then $K_p \cap K_q \subseteq K_r$, and hence $A_S = \emptyset$ whenever $p, q \in S$ and $r \not\in S$.
There are at least $\frac{1}{8}2^n$ sets $S \subseteq Q_n$ with this property.
Thus $L$ has at most $\frac{7}{8}2^n$ atoms, as required.

Otherwise, $\cR$ is exactly $\{(p, q, p) \mid p, q \in Q_n\} \cup \{(p, q, q) \mid p, q \in Q_n\}$; that is, it has no extraneous triples beyond those required by (A) and (B).
Here $\tle$ is the partial order defined by $p \tle q$ if and only if $q = 0$.
By Condition 2, every transformation $t \in T_{\cD}$ is monotone with respect to this order.
Thus, if $0t \not= 0$, then $Q_nt = 0t$.
If $0 \not\in F$ this implies that $K_0 \subseteq K_p$ for all $p \in Q_n$, since $0t \in F \implies Q_nt = 0t \in F$.
If $0 \in F$, we discover that $K_p \subseteq K_0$ for all $p \in Q_n$, since $0t \not\in F \implies Q_nt = 0t \not\in F$.
Hence, there must be some containment between quotients of $L$, say $K_p \subseteq K_q$ for $p \not= q$.
Then $A_S = \emptyset$ whenever $p \in S$ and $q \not\in S$. Thus, $L$ has at most $\frac{3}{4}2^n$ atoms. \qed
\end{proof}

The proof of Theorem~\ref{thm:reversalbd} actually tells us a great deal about what a complex witness must look like.
Extending the proof slightly, we obtain an important corollary.
\begin{corollary}
\label{cor:reversalbd}
Suppose $L$ is a suffix-convex language with $\kappa(L) = n \ge 3$ and $\kappa(L^R) = 2^n - 2^{n-3}$.
Let $\cD = (Q_n, \Sigma, \delta, 0, F)$ be a minimal DFA for $L$ and assume that $\cD$ respects a triple system $\cR$.
Then there are exactly two non-zero states $p$ and $q$ such that $p \tle_\cR q$.
Furthermore, $\tle_\cR$ is a partial order, i.e. no states are symmetric with respect to $\tle_\cR$.
\end{corollary}
\begin{proof}
We prove that $\tle_\cR$ is a partial order first:
To a contradiction, suppose there are distinct states $p$ and $q$ such that $p \sim q$.
Then $K_0 \cap K_p \subseteq K_q$ and $K_0 \cap K_q \subseteq K_p$.
If $p = 0$, then this implies that $K_0 \subseteq K_q$ and hence $L$ has at most $\frac{3}{4}2^n$ atoms.
We have a similar contradiction if $q = 0$.
Finally, if $p$ and $q$ are both non-zero, then $A_S = \emptyset$ whenever $0, p \in S$ and $q \not\in S$ or $0, q \in S$ and $p \in S$; again, we conclude that $L$ can have no more than $\frac{3}{4}2^n$ atoms. This contradicts the assumption $\kappa(L^R) = \frac{7}{8}2^n$.

Now consider the first claim. Suppose there are no non-zero states $p$ and $q$ with $p \tle q$
Then $\tle$ is the partial order defined by $p \tle q$ if and only if $q = 0$, exactly as in the last case of Theorem~\ref{thm:reversalbd}.
The same argument shows that $L$ has at most $\frac{3}{4}2^n$ atoms, a contradiction.

Second, suppose there are at least two pairs of distinct non-zero states, say $(p_1, q_1)$ and $(p_2, q_2)$ such that $p_1 \tle q_1$ and $p_2 \tle q_2$.
We have
\begin{equation}
0, p_1 \in S \text{ and } q_1 \not\in S \implies A_S = \emptyset.
\end{equation}
This reduces the number of atoms of $L$ to at most $\frac{7}{8}2^n$.
Similarly,
\begin{equation}
0, p_2 \in S \text{ and } q_2 \not\in S \implies A_S = \emptyset.
\end{equation}
If $p_1$, $q_1$, $p_2$, and $q_2$ are all distinct, (1) and (2) reduce the number of atoms to at most $\frac{3}{4}2^n$.
However, even if there is some overlap between $(p_1, q_1)$ and $(p_2, q_2)$ (e.g. $p_1 = p_2$ or $q_1 = p_2$),
we can still deduce that the number of atoms is less than $\frac{7}{8}2^n-1$, yielding a contradiction. A more thorough case analysis shows that that number of atoms is at most $\frac{13}{16}2^n$. \qed
\end{proof}

Simply put, this corollary says that the triple system of a witness for reversal must look something like Figure~\ref{fig:revsystem}.
We use this triple system to create our witness for reversal.

\begin{definition}
\label{def:revsystem}
For $n \ge  3$, define $\cS_n = (Q_n, 0, \{1\}, \cR_n)$ where
\[\cR_n = \{(p, q, p) \mid p, q \in Q_n\} \cup \{(p, q, q) \mid p, q \in Q_n\} \cup \{(0,2, 1), (2, 0, 1)\}.\]
\end{definition}

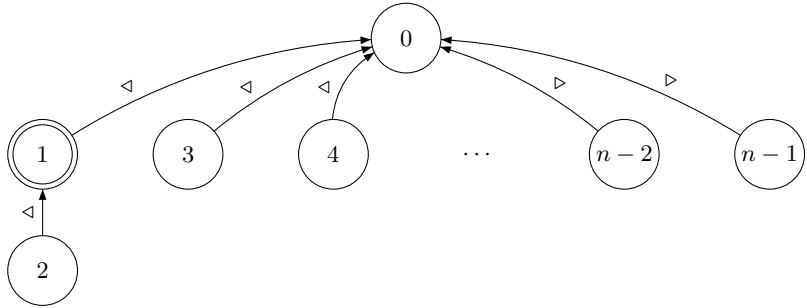
\begin{figure}
\unitlength 11pt
\begin{center}\begin{picture}(25,8)(0,0)
\gasset{Nh=2.4,Nw=2.4,Nmr=1.2,ELdist=0.3,loopdiam=1.2}
\node(0)(12.5,8){$0$}
\node(1)(0,4){$1$}\rmark(1)
\node(2)(0,0){$2$}
\node(3)(5,4){$3$}
\node(4)(10,4){$4$}
\node[Nframe=n](5)(15,4){$\cdots$}
\node(6)(20,4){$n-2$}
\node(7)(25,4){$n-1$}

\drawedge(2,1){$\tl$}
\drawedge[curvedepth=1,ELpos=27](1,0){$\tl$}
\drawedge[curvedepth=0.6,ELpos=34](3,0){$\tl$}
\drawedge[curvedepth=1,ELpos=40](4,0){$\tl$}
\drawedge[curvedepth=-0.6,ELpos=37,ELside=r](6,0){$\tg$}
\drawedge[curvedepth=-1,ELpos=31,ELside=r](7,0){$\tg$}

\end{picture}\end{center}
\caption{The order relation $\tle_{\cR_n}$ of Definition~\ref{def:revsystem} used in the complex witness stream for reversal.}
\label{fig:revsystem}
\end{figure}

Notice that $\tle_{\cR_n}$ is a partial order and $\cR_n$ is the relation defined in Theorem~\ref{thm:mono}.
Hence there exists a DFA respecting $\cS_n$ whose transition semigroup contains every transformation of $Q_n$ that is monotone with respect to $\tle_{\cR_n}$.
However, the DFA defined in that theorem has an enormous alphabet, with one letter for each monotone function.
Instead, we define a trimmed-down version with a smaller transition semigroup to be our witness for reversal.

\begin{definition}
\label{def:revwitness}
For $n \ge 3$, let $L_n(\Sigma)$ be the language recognized by the DFA $\cD_n = (Q_n, \Sigma, \delta_n, 0, \{1\})$, where $\Sigma = \{a, b, c, d, e, f, g, h\}$ and $\delta_n$ is given by the transformations \begin{align*}
a &\colon (3, 4, \dots, n-1) & e &\colon (1 \to 2)\\
b &\colon (3 \to 1) & f &\colon (2 \to 1)\\
c &\colon (3 \to 2) & g &\colon (Q_n \to 3)\\
d &\colon (1 \to 0) & h &\colon (Q_n\setminus\{0\} \to 2)(0 \to 1)
\end{align*}
\end{definition}

\begin{figure}
\unitlength 11pt
\begin{center}\begin{picture}(27,6.5)(-6,1.5)
\gasset{Nh=2.4,Nw=2.4,Nmr=1.2,ELdist=0.3,loopdiam=1.2}
\node(0)(-5,4){$0$}\imark(0)
\node(1)(0,6){$1$}\rmark(1)
\node(2)(0,2){$2$}
\node(3)(5,4){$3$}
\node(4)(10,4){$4$}
\node[Nframe=n](5)(15,4){$\cdots$}
\node(n-1)(20,4){$n-1$}

\drawedge(3,4){$a$}
\drawedge(4,5){$a$}
\drawedge(5,6){$a$}
\drawedge[curvedepth=2](n-1,3){$a$}

\drawedge[curvedepth=1](1,2){$e$}
\drawedge[curvedepth=1](2,1){$f$}

\drawedge[curvedepth=-1,ELside=r](3,1){$b$}
\drawedge[curvedepth=1](3,2){$c$}

\drawedge[curvedepth=-1,ELside=r](1,0){$d$}
\end{picture}\end{center}
\caption{DFA $\cD_n$ of Definition~\ref{def:starwitness}. Missing transitions are self-loops. Letters $g$ and $h$ not shown.}
\label{fig:revwitness}
\end{figure}
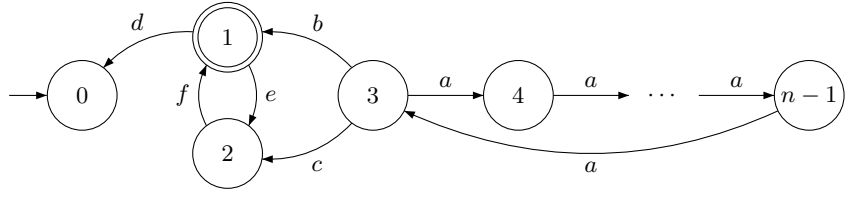

\begin{proposition}
\label{prop:revproper}
For $n \ge 3$, $L_n(\Sigma)$ of Definition~\ref{def:revwitness} is proper suffix-convex and $\kappa(L_n) = n$.
\end{proposition}
\begin{proof}
DFA $\cD_n(\Sigma)$ is minimal because every state is reachable by a word in $\{a, b, c, g\}^*$ and any two states are distinguished by a word in $\{a, b, f\}^*$.
By Theorem~\ref{thm:mono}, $L_n$ must be suffix-convex because $\cD_n$ respects the triple system $\cS_n$ of Definition~\ref{def:revsystem}.
It cannot be a left ideal because $h \in L_n$ but $h^2 \not\in L$, it cannot be suffix-closed because $\eps \not\in L_n$, and it is not suffix-free because $bh \in L_n$ and $h \in L_n$; thus $L_n$ is a proper language.\qed
\end{proof}

\begin{theorem}
\label{thm:reversal}
The language stream $(L_n(\Sigma) \mid n \ge 3)$ of Definition~\ref{def:revwitness} meets the upper bound for reversal of proper suffix-convex languages.
That is, for $n \ge 3$, $\kappa(L_n^R) = 2^n - 2^{n-3}$.
\end{theorem}
\begin{proof}
We obtain an NFA recognizing $L_n^R$ by reversing every transition in $\cD_n$ and interchanging the initial state 0 with the final state 1.
That is, we have an NFA $\cN = (Q_n, \Sigma, \delta^R, \{1\}, \{0\})$ where $p \in \delta^R(q, \ell) \iff \delta_n(p, \ell) = q$ for all $\ell \in \Sigma$.
Extending this relation to words, we have $p \in \delta^R(q, w) \iff \delta_n(p, w^R) = q$ for all $w \in \Sigma^*$.
We prove that every state $S \subseteq Q_n$ such that $0, 2 \in S \implies 1 \in S$ is reachable in this NFA, and that these states are pairwise distinguishable.
There are four cases for reachability.

\textbf{Case 1.} $0 \not\in S$ and $1 \in S$: Let $w$ be a word that induces $(S \to 1)$ in $\cD_n$; such a word can be found by using $a$ and $b$ to map states in $S \cap \{3, 4, \dots, n-1\}$ to $1$, and then appending $f$ if $2 \in S$. The reversal of $w$ maps the initial state $1$ to $S$ in $\cN$.

\textbf{Case 2.} $0 \not\in S$ and $1 \not\in S$: Use a word that induces $(1 \to 0)(S \to 1)$ in $\cD_n$. This word is given by $dw$, where $w$ is found as in Case 1.

\textbf{Case 3.} $0 \in S$ and $1 \in S$: Use a word $w$ that induces $(S \to 1)(\overline{S} \to 2)$ in $\cD_n$.
This word can be constructed as the concatenation of two parts.
For the first part, begin with $f$ if $2 \in S$, and then use $a$ and $b$ to send the states of $S \cap \{3, 4, \dots, n-1\}$ to $1$.
The second part is simply $dh$.
The first part maps $S \setminus \{0\}$ to $1$, and it leaves all the states of $\overline{S}$ in $Q_n \setminus \{0, 1\}$.
Further applying $dh$ maps $\{0,1\}$ to $1$ and the remaining states to $2$.
Thus, $w$ induces $(S \to 1)(\overline{S} \to 2)$ in $\cD_n$.
The reverse of $w$ maps $1$ to $S$ in $\cN$.

\textbf{Case 4}, $0 \in S$, $1 \not\in S$, and $2 \not\in S$: Use a word that induces $(S \to 1)(\overline{S} \to 2)$ in $\cD_n$, which is given by $ew$, where $w$ is constructed as in Case 3.

No other cases are possible, since we are only interested in sets $S$ where $0, 2 \in S \implies 1 \in S$.

For distinguishability, observe that for each $p \in Q_n\setminus\{0\}$, the transformation $(Q_n \to p)$ is induced in $\cD_n$ by a word $w_p$ in $\{a, b, c, g, h\}^*$.
If $S_1, S_2 \subseteq Q_n$ are distinct sets of states with $p \in S_1 \oplus S_2$, they are distinguished in $\cN$ by $w_p^R$ if $p \not= 0$.
If $p = 0$, then they are distinguished by $\eps$.
Therefore $\kappa(L_n) = 2^n - 2^{n-3}$. \qed
\end{proof}


\section{Syntactic Semigroup}

The final complexity measure we consider is \emph{syntactic complexity}, the size of the syntactic semigroup.
The size and nature of the most complex semigroup in the class of proper suffix-convex languages is an interesting and difficult open question.
We describe a stream that we conjecture to be maximal in this respect.

\begin{definition}
For $n \ge  3$, define $\cS_n = (Q_n, 0, \{n-2\}, \cR_n)$ where
\begin{align*}
\cR_n =  \{&(p, q, p) \mid p, q \in Q_n\} \cup \{(p, q, q) \mid p, q \in Q_n\}\\
\cup &\{(0, p, q) \mid 0 \le p, q \le n-2\} \cup \{(p, 0, q) \mid 0 \le p, q \le n-2\}\\
\cup &\{(0, n-1, q) \mid q \le n-2\} \cup \{(n-1, 0, q) \mid q \le n-2\}.
\end{align*}
\end{definition}

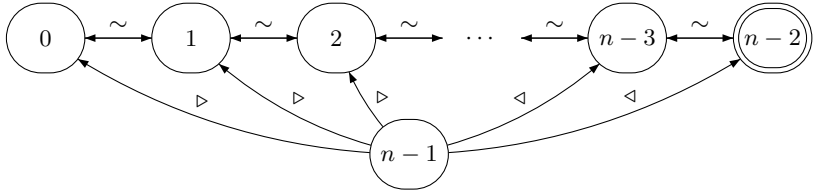
\begin{figure}
\unitlength 11pt
\begin{center}\begin{picture}(25,4)(0,0)
\gasset{Nh=2.4,Nw=2.7,Nmr=1.2,ELdist=0.3,loopdiam=1.2}
\node(0)(0,4){$0$}
\node(1)(5,4){$1$}
\node(2)(10,4){$2$}
\node[Nframe=n](3dots)(15,4){$\cdots$}
\node(n-3)(20,4){$n-3$}
\node(n-2)(25,4){$n-2$}\rmark(n-2)
\node(n-1)(12.5,0){$n-1$}

\drawedge[curvedepth=1,ELpos=56,ELside=r](n-1,0){$\tg$}
\drawedge[curvedepth=0.66,ELpos=50,ELside=r](n-1,1){$\tg$}
\drawedge[curvedepth=0.33,ELpos=45,ELside=r](n-1,2){$\tg$}
\drawedge[curvedepth=-0.66,ELpos=50](n-1,n-3){$\tl$}
\drawedge[curvedepth=-1,ELpos=60](n-1,n-2){$\tl$}

\drawedge(0,1){$\sim$}
\drawedge(1,0){}
\drawedge(1,2){$\sim$}
\drawedge(2,1){}
\drawedge(2,3dots){$\sim$}
\drawedge(3dots,2){}
\drawedge(3dots,n-3){$\sim$}
\drawedge(n-3,3dots){}
\drawedge(n-3,n-2){$\sim$}
\drawedge(n-2,n-3){}
\end{picture}\end{center}
\caption{The order relation $\tle_\cR$ used in the conjectured witness for syntactic semigroup.}
\label{fig:synwitness}
\end{figure}

We can compute a bound on the number of transformations that respect $\cR_n$. Observe that:
\begin{enumerate}[topsep=0pt]
\item If $0$ is fixed by $t$, then $pt \not= n-1$ for all $p \le n-2$ by monotonicity.
\item If $0t \in \{1, 2, \dots, n-2\}$ then $\{1, 2, \dots, n-2\}t = 0t$ since $(0t, 0t, pt)$ must be in $\cR$ for all $p \in Q_n \setminus \{n-1\}$, and this fails unless $pt = 0t$.
\item If $0t = n-1$ then $Q_nt = n-1$ by monotonicity.
\end{enumerate}
By 1, the number of transformations satisfying $0t = 0$ is at most $n (n-1)^{n-2}$.
By 2, the number of transformations satisfying $0t \in \{1, 2, \dots, n-2\}$ is at most $n (n-2)$.
By 3, there is only one transformation where $0t = n-1$.
Thus, the size of the transition semigroup is at most $n(n-1)^{n-2} + (n-1)^2$.
A more careful analysis reveals that every transformation counted by this argument satisfies Conditions 1 and 2,
and thus they all may be added to the transition semigroup.\footnote{This fact is offered without proof, but it is not difficult to verify.}
There is a fairly simple DFA that respects $\cS$ and has this semigroup:

\begin{definition}
\label{def:syntacticwitness}
For $n \ge 3$, let $L_n(\Sigma)$ be the language recognized by the DFA $\cD_n = (Q_n, \Sigma, \delta_n, 0, \{n-2\})$, where $\Sigma = \{a, b, c, d, e, f, g, h\}$ and $\delta_n$ is given by the transformations 
\begin{align*}
a &\colon (1, \dots n-2), & e &\colon (Q_n \setminus \{n-1\} \rightarrow 1),\\
b &\colon (1, 2), & f &\colon (n-1 \rightarrow 0),\\
c &\colon (n-2 \rightarrow 1), & g &\colon (n-1 \rightarrow 1),\\
d &\colon (n-2 \rightarrow 0), & h &\colon (Q_n \rightarrow n-1).
\end{align*}
\end{definition}

The transition semigroup of $\cD_n$ contains every transformation satisfying Condition 1 and Condition 2 with respect to $\cS_n$ of Definition~\ref{def:syntacticwitness}. Hence the size of the syntactic semigroup of $L_n$ is $n(n-1)^{n-2} + (n-1)^2$. We conjecture that this is optimal.

\begin{conjecture}
\label{conj:syntactic}
For $n \ge 3$, the syntactic complexity of any proper suffix-convex language of complexity at most $n$ is at most $n(n-1)^{n-2} + (n-1)^2$.
\end{conjecture}

Note that the conjectured bound does not hold for general suffix-convex languages, as there is a left ideal stream with syntactic complexity $n^{n-1} + n-1$~\cite{BrYe11}.
This witness is known to have maximal syntactic complexity among left ideals and suffix-closed languages. If Conjecture~\ref{conj:syntactic} holds, it would imply that the left ideal witness has the largest syntactic complexity over all suffix-convex languages, since suffix-free languages are known to have smaller syntactic complexity~\cite{BrSz15a}.

\section{Most Complex Streams}

A \emph{most complex} language stream is required to meet all the operational bounds for reversal, star, product, and boolean operations, as well as the bound for syntactic complexity.\footnote{Additionally, it is usually required that the atoms of the language are as complex as possible~\cite{Brz13}, but this measure is not discussed here.}
Using results already stated in this paper, we can easily show that there is no most complex proper stream.

\begin{lemma}
\label{lem:twowit}
For $n \ge 4$, there does not exist a proper suffix-convex language of complexity $n$ that meets the complexity bounds for both reversal and star.
\end{lemma}
\begin{proof}
Suppose $L$ is a proper suffix-convex language of complexity $n$ with $\kappa(L^R) = 2^n - 2^{n-3}$ and $\kappa(L^*) = 2^{n-1} + 2^{n-2}$.
Let $\cD = (Q_n, \Sigma, \delta, 0, F)$ be a minimal DFA for $L$.
By Corollary~\ref{cor:reversalbd}, any triple system $\cR$ respected by $\cD$ must have only two states (besides 0) that are comparable by $\tle_\cR$.
Yet by Lemma~\ref{lem:starreq}, every pair of states $(p, q)$ must have some comparison in $\tle_\cR$ (either $p \tle q$ or $q \tle p$).
This is impossible for $n \ge 4$.
\qed
\end{proof}

Surprisingly, even though the true upper bound for syntactic complexity is not known with certainty, we can still prove that a third stream, different from those for reversal and star, is needed to meet this bound.
Thus, at least three streams are needed to meet all the bounds.
\begin{theorem}
\label{thm:threewit}
For $n \ge 4$, there does not exist a proper suffix-convex language of complexity $n$ that meets the upper bounds for any two of reversal, star, and syntactic complexity.
\end{theorem}
The  proof of this theorem is considerably longer and more difficult than the previous proofs. See Appendix~\ref{app:threewit}.

\section{Conclusion}

We have exhibited several new tight upper bounds for proper suffix-convex languages, some of which apply to all suffix-convex languages.
The introduction of triple systems was an essential tool in this endeavour, so perhaps variant triple systems can be developed for other difficult classes of regular languages.
The question of determining the maximal syntactic complexity of proper languages remains open. It may be solvable using a similar approach to the proof of Theorem~\ref{thm:threewit}.

\section*{Acknowledgement}
This work arose from a fruitful collaboration with Janusz Brzozowski, without which it would not have been possible. I am extremely grateful for his guidance and mentorship.

\bibliography{SC}
\bibliographystyle{splncs03}


\newpage
\appendix

\section{Proof of Theorem~\ref{thm:threewit}}
\label{app:threewit}

\begin{rethm}{Theorem~\ref{thm:threewit}}
{For $n \ge 4$, there does not exist a proper suffix-convex language of complexity $n$ that meets the upper bounds for any two of reversal, star, and syntactic complexity.}
\end{rethm}
\begin{proof}
From the previous sections, we know that the bound for reversal is $2^n - 2^{n-3}$, the bound for star is $2^{n-1} + 2^{n-2}$, and the bound for syntactic complexity is at least $\Phi(n) = n(n-1)^{n-2} + (n-1)^2$ (though it could be larger if Conjecture~\ref{conj:syntactic} is not true).
Lemma~\ref{lem:twowit} shows that no suffix-convex language can meet the bounds for both reversal and star.

By Corollary~\ref{cor:reversalbd}, every suffix-convex language that meets the reversal bound respects a triple system with a very restricted structure.
In particular, the syntactic semigroup of such a language can be no larger than the number of monotone transformations on $\tle_\cR$ of Definition~\ref{def:revsystem}.
Some careful counting shows that there exactly $2n^{n-2} + 3 \cdot 2^{n-3} + n-2$ such functions.
One can verify by computation that $2n^{n-2} + 3 \cdot 2^{n-3} + n-2 < \Phi(n)$ for $n \in \{4, 5, 6, 7\}$. If $n \ge 8$ then
\begin{align*}
2n^{n-2} + 3 \cdot 2^{n-3} + n-2 &\le 2\left(1+\frac{1}{n-1}\right)^{n-1}\frac{(n-1)^{n-1}}{n} + 2^{n-1}\\
&\le 2e\frac{(n-1)^{n-1}}{n} + 2^{n-1}\\
&\le (2e + 2)(n-1)^{n-2}\\
&\le n(n-1)^{n-2} < \Phi(n).
\end{align*}
In the second step above we used the fact that $\left(1+\frac{1}{x}\right)^x$ increases to $e$ as $x \to \infty$.
Thus, for $n \ge 4$, any witness for reversal cannot be a witness for syntactic complexity.
It remains to prove that no proper language can meet the bounds for both star and syntactic semigroup.

Let $L$ be a proper suffix-convex language of complexity $n$ with $\kappa(L^*) = 2^{n-1} + 2^{n-2}$ and let $\cD = (Q_n, \Sigma, \delta, 0, \{f\})$ be a minimal DFA for $L$.
We can be sure that $\cD$ has only one final state, since this is true of any regular language that meets the bound for star.
Lemma~\ref{lem:starreq} states that $\cD$ respects a triple system $\cS = (Q_n, 0, \{f\}, \cR)$ such that $p \tle_\cR q$ or $q \tle_\cR p$ for all $p, q \in Q_n$.
We can further restrict the structure of $\cS$.

\begin{lemma}
\label{lem:noppq}
If $n \ge 3$, then $(p, p, q) \not\in \cR$ for all distinct $p, q \in Q_n$.
\end{lemma}
\begin{proof}
Observe that \[w^{-1}L^* = \bigcup_{\substack{w = uv\\ u \in L^*}} v^{-1}L^*.\]
In particular, every quotient of $L^*$ is a union of quotients of $L$ concatenated with $L^*$.
We may therefore write a quotient of $L^*$ in the form $J_S \coloneqq \bigcup_{p \in S} K_p L^*$ for $S \subseteq Q_n$.
Since $\eps \in K_f$ and $K_0 = L$, we have $K_0L^* \subseteq K_fL^*$.
Hence if $f \in S$ then $J_{S \setminus \{0\}} = J_{S \cup \{0\}}$.
This proves that there can only be $\frac{3}{4}2^n = 2^{n-1} + 2^{n-2}$ distinct quotients of $L^*$ (for all regular languages, not just suffix-convex ones).
Moreover, observe that if there is any containment between quotients other than $K_0 \subseteq K_f$, then this bound is reduced below $2^{n-1} + 2^{n-2}$.
Thus $(p, p, q) \not\in \cR$ for all $p \not= q$, except possibly when $p=0$ and $q = f$.

Now suppose to a contradiction that $(0, 0, f) \in \cR$.
Then $0 \sim f$, and no other state is symmetric with $0$ or $f$.
Perform the NFA construction for star on $\cD$: Add a new initial and final state $0'$ to $\cD$ with the same transitions as $0$, and add an $\eps$-transition from $f$ to $0$.
Since $L$ meets the bound for star, all of the following sets must be reachable in this $\eps$-NFA.
\[\{\{0'\}\} \cup \{P \subseteq Q_n \mid P \not= \emptyset \text{ and } (f \in P \implies 0 \in P)\}\]
However, this implies that some states are not reachable in the $\eps$-NFA:
We prove inductively that the states in any reachable set are pairwise symmetric.
This is certainly true of the initial state $\{0\}$. By monotonicity, any set of symmetric states in $\cD$ can only be mapped to another set of symmetric states.
The only way for the reachable sets to grow in size is by mapping onto a set containing $f$, at which point $0$ is added by the $\eps$-transition, however this cannot violate the symmetry between states because $0 \sim f$.
Since $n \ge 3$ there is some state that is not symmetric with $0$, and hence not all states are reachable.
This contradicts the complexity of $L^*$ and proves the lemma.\qed
\end{proof}

We have now established that $\tle$ must look like an ordered sequence of ``pods'' of pairwise symmetric elements, and $0$ must be in a pod by itself.
Let $P_0, P_1, \dots, P_\ell$ be the pods, where $P_0 = \{0\}$ and $p \tg q$ if and only if $p \in P_i$ and $q \in P_j$ with $i < j$.
Without loss of generality, relabel the states so that $0 \tge_\cR 1 \tge_\cR 2 \tge_\cR \cdots \tge_\cR n-2 \tge_\cR n-1$;
hence $P_1 = \{1, 2, \dots, |P_1|\}$, $P_2 = \{|P_1|+1, \dots, |P_1| + |P_2|\}$, $\dots$, and $P_\ell = \{n-|P_\ell|, \dots, n-1\}$.
Let $m = \max_{1 \le i \le \ell} |P_i|$, and let $i^* \in \{1, \dots, \ell\}$ such that $|P_{i^*}| = m$.
We have four cases based on the value of $m$.

\noindent\textbf{Case 1: $m = 1$.}\\
If $m=1$ there is no pair of symmetric elements. Then $\tle$ is a total order, just as in the system of Definition~\ref{def:starsystem}.
Here, every transformation in $T_\cD$ must monotone with respect to the total order.
It is well known, and not too hard to show, that there are exactly ${2n-1}\choose n$ functions on $Q_n$ that are monotone with respect to the total order $\le$.
We must show ${{2n-1} \choose n} < \Phi(n)$ for $n \ge 4$.
This is easily checked by hand for $n = 4, 5, 6$. If $n \ge 7$, then we have
\[ {{2n-1} \choose n} \le \frac{(2n)!}{(2n) n!(n-1)!} \le \frac{(2^n n!)^2}{2 n! n!} = \frac{1}{2}4^n \le 6^{n-1} < \Phi(n).\]

\noindent\textbf{Case 2: $2 \le m \le n-2$.}\\
We estimate the number of possible transformations in $T_\cD$ in three groups.

\noindent\emph{Group 1:} $\{t \in T_\cD \mid 0t = 0, pt \not=0 \text{ for all } p \not= 0\}$.\\
This group contains the transformations $t$ that fix $0$ and do not map any other state to $0$.
By monotonicity, every transformation must map $P_{i^*}$ into some pod $P_j$. Even if there are no restrictions on how the remaining $n-1-m$ states can be mapped, this observation gives the upper bound of $(n-1)^{n-1-m}\sum_{j=1}^\ell |P_j|^m$ transformations in this group.
The sum in this expression is maximized when the states of $Q_n$ are concentrated into as few pods as possible, since if $|P_i| \le |P_j|$ then $|P_i|^m + |P_j|^m < (|P_i| - \eps)^m + (|P_j + \eps)^m$ for any $\eps > 0$.
No pod can have size larger than $m$; hence this maximum occurs when as many pods as possible have size $m$ and the left over states are put into a single pod.
Thus,
\begin{align*}
(n-1)^{n-m}\sum_{j=1}^\ell |P_j|^m \le &(n-1)^{n-1-m}\left(\left\lfloor \frac{n-1}{m}\right\rfloor m^m + \left(\left(\frac{n-1}{m} - \left\lfloor\frac{n-1}{m}\right\rfloor\right) m \right)^m\right)\\
 \le &(n-1)^{n-1-m}\left(\frac{n-1}{m}\right) m^m\\
 = &(n-1)^{n-m}m^{m-1}.
\end{align*}

\noindent\emph{Group 2:} $\{t \in T_\cD \mid 0t = 0, pt =0 \text{ for some } p \not= 0\}$.\\
Here we have all transformations that fix $0$ and also map some non-zero state to $0$.
By monotonicity, if any state $p$ in some pod $P_j$ is mapped to $0$ by $t$, then $P_jt = \{0\}$, and in fact $P_it = \{0\}$ for all $i \le j$. Thus, the number of transformations $t$ that map some non-zero state to $0$ is at most
\[\sum_{j=1}^\ell (n-1)^{n-1 - \sum_{1 \le i \le j} |P_i|}.\]
This expression is maximized when $|P_1| = |P_2| = \cdots = |P_{\ell-1}| = 1$, and $|P_\ell| = n-\ell$.
In this case, $P_\ell$ must have size $m$ since all other pods have size $1$.
Thus, it is bounded by \[\sum_{j=1}^{\ell-1} (n-1)^{n-1-j} + 1 \le (n-1)^{n-2} + (n-1)^{n-3} + \cdots + (n-1)^m + 1.\] For $n \ge 4$, this sum is bounded by $\frac{3}{2} (n-1)^{n-2}$.

\noindent\emph{Group 3:} $\{t \in T_\cD \mid 0t \not= 0\}$.\\
Finally, we have the group containing transformations that do not fix $0$.
This group requires some deeper analysis than just using monotonicity.
Naively, there can only be $(n-1)^n$ transformations in this group, however this is too much of an overestimate.
We reduce this using the fact that $(p, p, q) \not\in \cR$ for all distinct $p, q \in Q_n$.

Let $t$ be a transformation in this group, and consider the sequence $0t, 0t^2, 0t^3, \dots, 0t^n$.
Since these $n$ states all lie in $Q_n \setminus \{0\}$, the sequence must enter a cycle at some point and repeat a state.
Choose $i \in \{1, \dots, n-1\}$ as small as possible such that $0t^n = 0t^i$.
By monotonicity, and since $\tle$ admits a comparison between each pair of states, we have $0t^y \tle 0t^x$ whenever $x \le y$. In particular, if $j \in \{i+1, \dots, n-1\}$, then $(0, 0t^{n-i}, 0t^{j-i}) \in \cR$.
By Condition 1, it follows that $(0t^i, (0t^{n-i})t^i, (0t^{j-i})t^i) = (0t^i, 0t^i, 0t^j) \in \cR$.
By Lemma~\ref{lem:noppq}, it must be that $0t^i = 0t^j$.
Thus, $0t^i = 0t^{i+1} = 0t^{i+2} = \cdots = 0t^n$.
In other words, $0t^i$ is a fixed point of $t$.

Now let $p$ be a state with $p \tle 0t^i$.
Since $(0, 0t^i, p) \in \cR$ and $0t^{i+1} = 0t^i$, we have $(0t^i, 0t^i, pt^i) \in \cR$ by Condition 1, and hence $pt^i = 0t^i$.
This implies that every state below $0t^i$ is mapped to state $0t^i$ by sufficiently large powers of $t$.
If $0t^i \tle p \tle$, then by monotonicity $0t^i \tle pt^i \tle 0t^i$; hence, $pt^i \sim 0t^i$.
Then $(0, 0t^i, pt^i) \in \cR$, which implies that $(0t^i, 0t^{2i}, pt^{2i}) = (0t^i, 0t^i, pt^{2i}) \in \cR$ by Condition 1, and thus $pt^{2i} = 0t^i$.
Therefore, $t$ has exactly one fixed point, and every state must eventually be mapped to this state by a sufficiently large power of $t$.

Let us count the number of transformations with this property.
Suppose $r \in Q_n \setminus \{0\}$ is the fixed point of $t$.
Construct a directed graph on $Q_n\setminus \{0\}$ by adding an edge $(p, pt)$ for each state $p \not= r$.
This graph encodes the behaviour of $t$ on $Q_n \setminus \{0\}$.
Every vertex except $r$ has out-degree 1, and since there is some constant $k$ such that $pt^k = r$ for all $p \in Q_n \setminus \{0\}$, the graph contains no cycles.
Thus, the graph is a directed tree rooted at $r$.

The number of such trees is well-known.
By Cayley's formula, there are exactly $k^{k-2}$ labelled, unordered trees on $k$ elements\footnote{\emph{Labelled} means that the vertices of the trees are labelled with the integers from $1$ to $k$. \emph{Unordered} means that the neighbours of a vertex have no particular order.}.
From this, we may deduce that there are exactly $k^{k-1}$ labelled, unordered, and rooted trees on $k$ elements, simply because each unrooted tree can be made into $k$ different rooted trees depending on which vertex is chosen to be the root.

We established that the behaviour of $t$ on $Q_n \setminus \{0\}$ can be represented as such a graph.
Thus, there are at most $(n-1)^{n-2}$ ways for $t$ to act on $Q_n \setminus \{0\}$.
Assuming that $0$ can be mapped anywhere in $Q_n \setminus \{0\}$, we conclude that there are at most $(n-1)^{n-1}$ transformations in this group.

This bound is close, but not quite sufficient for our purposes.
If the fixed point of some transformation $t$ in this group lies in $P_1$, then $t$ cannot map any state into $P_\ell$, for then $P_\ell t\subseteq P_\ell$ and some states will never reach the fixed point of $t$.
If the fixed point does not lie in $P_1$, then no state can be mapped into $P_1$ by monotonicity.
In either case, there is some state $p \not= 0$ that cannot be in the image of $t$, and $p$ depends only on the fixed point.
As above, there are at most $(n-2)^{n-3}$ possible behaviours of $t$ on $Q_n \setminus \{0, p\}$.
Hence we have the bound $(n-2)^{n-1}$ as $0$ and $p$ might be mapped to any states in $Q_n \setminus \{0, p\}$.

\noindent\emph{Combining Groups:}\\
Adding the three estimates together, we obtain the bound
\[|T_\cD| \le (n-2)^{n-1} + m^{m-1}(n-1)^{n-m} + \frac{3}{2} (n-1)^{n-2}.\]

We prove that this bound is strictly less than $\Phi(n)$ for $n \ge 4$ and $2 \le m \le n-2$.
One can verify this statement by computation for $4 \le n \le 16$ and $2 \le m \le n-2$.

By calculus, $\left(1 - \frac{1}{n-1}\right)^{n-1}$ is increasing with $n$, and it is known to converge to $\frac{1}{e} < \frac{2}{5}$.
Hence $(n-2)^{n-1} = \left(1 - \frac{1}{n-1}\right)^{n-1} (n-1)^{n-1} \le \frac{2}{5} (n-1)^{n-1}$.

Next, we prove $(n-1)^{n-m}m^{m-1} \le \frac{1}{2}(n-1)^{n-1}$.
Observe $(n-1)^{n-m}m^{m-1} = \left(\frac{m-1}{n-1}\right)^{m-1} (n-1)^{n-1}$.
With $x = m-1$ and $y = n-1$, $\left(\frac{m-1}{n-1}\right)^{m-1} = \left(\frac{x}{y}\right)^x$;
since $\frac{d}{dx} \left(\frac{x}{y}\right)^x = (1 + \ln x - \ln y) \left(\frac{x}{y}\right)^x$, the only local maximum or minimum of this function occurs at $x = y/e$.
Thus, $\left(\frac{m-1}{n-1}\right)^{m-1}$ is maximized either at $m = \frac{n-1}{e} + 1$ or at the endpoints $m = 2$ and $m = n-2$.
We compare the function values at these points: At $m = 2$, $\left(\frac{m-1}{n-1}\right)^{m-1} = \frac{1}{n-1} \le \frac{1}{2}$.
At $m = \frac{n-1}{e}+1$, $\left(\frac{m-1}{n-1}\right)^{m-1} = \left(\frac{1}{e}\right)^{\frac{n-1}{e}} \le \frac{1}{2}$.
At $m = n-2$, $\left(\frac{m-1}{n-1}\right)^{m-1} \le \left(1 + \frac{1}{n-2}\right)^2 \left(1 - \frac{1}{n-1}\right)^{n-1}(n-1)^{n-1}$.
As before, $\left(1 - \frac{1}{n-1}\right)^{n-1} \le \frac{2}{5}$.
For $n \ge 12$, $\left(1 + \frac{1}{n-2}\right)^2 \le \left(1 + \frac{1}{10}\right)^2 = 1.21 < \frac{5}{4}$.
Thus $\left(1 + \frac{1}{n-2}\right)^2 \left(1 - \frac{1}{n-1}\right)^{n-1}(n-1)^{n-1} \le  \frac{1}{2} (n-1)^{n-1}$.

Finally, $\frac{3}{2} (n-1)^{n-2} \le \frac{3}{2(n-1)} (n-1)^{n-1} \le \frac{1}{10} (n-1)^{n-1}$ for $n \ge 16$.
Therefore, for $n \ge 16$,
\[|T_\cD| \le \frac{1}{2}(n-1)^{n-1} + \frac{1}{10} (n-1)^{n-1} + \frac{2}{5}(n-1)^{n-1} = (n-1)^{n-1} < \Phi(n).\]
This proves the claim for the case $2 \le m \le n-2$.

\noindent\textbf{Case 3: $m = n-1$.}\\
In this case, there are only two pods, $P_0 = \{0\}$ and $P_1 = \{1, 2, \dots, n-1\}$.
We introduce a new property of $\cR$ to approach this case.
\begin{lemma}
\label{lem:hij}
There is an ordering $q_0, q_1, q_2, \dots, q_{n-1}$ of the states in $Q_n$ such that
$(q_h, q_i, q_j) \in \cR$ for all $h, i, j \in \{1, \dots, n-1\}$ with $h < i$ and $h < j$.
\end{lemma}
\begin{proof}
Consider the NFA construction for star on $\cD$.
Since $L$ meets the bound for star, all of the states listed in Lemma~\ref{lem:noppq} must be reachable; in particular, the set $Q_n$ must be reachable.
The size of a reachable set in the NFA can only be increased by using the $\eps$-transition from $f$ to $0$.
Hence, there is some reachable set $S \subseteq Q_n$ with size $n-1$ and a transformation $t$ such that $St =Q_n$ in the NFA.
In $\cD$, this implies that $t$ must map $S$ to $Q_n\setminus \{0\}$ exactly, so that the $\eps$-transition in the NFA fills in the $0$ state in the NFA.
In summary, this transformation $t$ has the property $0t \in P_1$ and $\im(t) = P_1$ in $\cD$.

For $h = 0, 1, \dots, n-1$ define $q_h = 0t^h$; we check that $(q_h, q_i, q_j) \in \cR$ whenever $h < i, j$ and that $q_1, \dots, q_{n-1}$ is a permutation of $1, 2, \dots, n-1$.
Obviously $q_0 = 0$, and no other $q_i$ is equal to $0$.
If $1 \le h < i, j \le m$, then we have $(0, q_{i-h}, q_{j-h}) \in \cR$ since all states in $P_1$ are symmetric.
By Condition 1, we have $(0t^h, q_{i-h}t^h, q_{j-h}t^h) = (q_h, q_i, q_j) \in \cR$ as required.

To a contradiction, suppose $\{q_0, \dots, q_{n-1}\}$ is not equal to $Q_n$.
Then there is some repeated state $q_h = q_i$ with $1 \le h < i \le n-1$, where $h$ is chosen as small as possible.
Since $|\im(t)| = n-1$, there can only be one pair of distinct states that $t$ maps to the same location; these are $q_{h-1}$ and $q_{i-1}$.
Thus, the states of $Q_n \setminus \{q_0, q_1, \dots, q_{n-1}\}$ must be permuted by $t$.
However, if $p \not\in \{q_0, q_1, \dots, q_{n-1}\}$, then $(0, q_{i-h}, p) \in \cR$.
It follows by Condition 1 that $(0t^h, q_{i-h}t^h, pt^h) = (q_h, q_h, pt^h) \in \cR$, and thus $pt^h = q_h$.
This contradicts the fact that $Q_n \setminus \{q_0, q_1, \dots, q_{n-1}\}$ is permuted by $t$.
Therefore $\{q_0, \dots, q_{n-1}\} = Q_n$.
\qed
\end{proof}

We use this lemma to drastically reduce the number of possible transformations that respect $\cS$.
Without loss of generality, relabel the states of $P_1$ so that $q_0 = 0, q_1 = 1, q_2 = 2, \dots, q_{n-1} = n-1$.


Suppose now that $t \colon Q_n \to Q_n$ respects $\cS$ and $pt = qt$ for some $p, q \in Q_n$ with $p < q$.
For any $r \in P_1$ such that $r > p$, $(p, q, r) \in \cR$ by Lemma~\ref{lem:hij} and hence $(pt, qt, rt) \in \cR$ by Condition 1.
By Lemma~\ref{lem:noppq}, it must be that $pt =qt = rt$; it follows that $pt = (p+1)t = (p+2)t = \cdots = (n-1)t$.
In other words, there is a unique state $p_0 \in Q_n$ such that states $0, 1, 2, \dots, p_0$ are all mapped to different locations by $t$, and $p_0, p_0+1, \dots, n-1$ are all mapped to the same state by $t$.
Observe that $t$ is determined by the $(p_0+1)$-tuple $(0t, 1t, 2t, \dots, p_0t)$, since all the remaining states must be mapped to $p_0t$.
Within this tuple, unless $t$ is the transformation $(Q_n \to 0)$, the entries $1t, 2t, \dots, p_0t$ must lie in $P_1$ by monotonicity.

Thus, the number of transformations that fix $0$ is at most $\sum_{p_0 = 0}^{n-1}p_0! {n-1 \choose p_0} = \sum_{p_0 = 0}^{n-1}\frac{(n-1)!}{(n-1-p_0)!}$.
Similarly, the number of transformations that map $0$ into $P_1$ is at most $\sum_{p_0 = 0}^{n-2}(p_0+1)!{n-1 \choose p_0+1} = \sum_{p_0 = 0}^{n-2}\frac{(n-1)!}{(n-2-p_0)!}$.
Hence, the total number of transformations that respect $\cS$ is at most
\begin{align*}
&\sum_{p_0=0}^{n-1} \frac{(n-1)!}{(n-1-p_0)!} + \sum_{p_0 = 0}^{n-2}\frac{(n-1)!}{(n-2-p_0)!}\\
\le & (n-1)! 2\sum_{i=0}^\infty \frac{1}{i!}\\
\le & 2e (n-1)!
\end{align*}
For $n \ge 4$, $2e(n-1)! \le 2e (1 \cdot 2) (n-1)^{n-3} < 12 (n-1)^{n-3} \le n(n-1)^{n-2}$.
Therefore $|T_\cD| < \Phi(n)$ as required.
\qed
\end{proof}

\end{document}